\documentclass[num-refs]{wiley-article}

\usepackage{fullpage,amsmath,tkz-base}  
\usepackage{cases}
\newcommand{\floor}[1]{\left \lfloor #1 \right \rfloor}

\usepackage{rustic} 
\usepackage[T1]{fontenc}

\usepackage{duerer} 
\usepackage[T1]{fontenc}

\usepackage{trajan}
\usepackage[T1]{fontenc}

\usepackage{wrapfig}
\usepackage{upgreek}

\usepackage{siunitx}
\usepackage{framed}

\usepackage{stackengine}

\usepackage{autobreak}
\allowdisplaybreaks

\def \B{\mathcal{B}}
\def \C{\mathcal{C}}

\def \E{\mathcal{E}}

\def \Q{\mathcal{Q}}

\def \S{\mathcal{S}}
\def \T{\mathcal{T}}

\def \X{\mathcal{X}}

\def \bB {\mathcal{B}}

\def \sT {\mathscr{T}}

\def \fc{\mathbf{c}}
\def \fc{\mathbf{c}}

\def \fu{\mathbf{u}}

\def \fx{\mathbf{x}}
\def \fy{\mathbf{y}}

\def \fY{\mathbf{Y}}

\def \f0{\mathbf{0}}

\definecolor{blau_1a}{RGB}{93,133,195}
\definecolor{blau_2a}{RGB}{0,156,218}
\definecolor{gruen_3a}{RGB}{80,182,149}
\definecolor{gruen_4a}{RGB}{175,204,80}
\definecolor{gruen_5a}{RGB}{221,223,72}
\definecolor{orange_6a}{RGB}{255,224,92}
\definecolor{orange_7a}{RGB}{248,186,60}
\definecolor{rot_8a}{RGB}{238,122,52}
\definecolor{rot_9a}{RGB}{233,80,62}
\definecolor{lila_10a}{RGB}{201,48,142}
\definecolor{lila_11a}{RGB}{128,69,151}

\definecolor{blau_1b}{RGB}{0,90,169}
\definecolor{blau_2b}{RGB}{0,131,204}
\definecolor{gruen_3b}{RGB}{0,157,129}
\definecolor{gruen_4b}{RGB}{153,192,0}
\definecolor{gruen_5b}{RGB}{201,212,0}
\definecolor{orange_6b}{RGB}{253,202,0}
\definecolor{orange_7b}{RGB}{245,163,0}
\definecolor{rot_8b}{RGB}{236,101,0}
\definecolor{rot_9b}{RGB}{230,0,26}
\definecolor{lila_10b}{RGB}{166,0,132}
\definecolor{lila_11b}{RGB}{114,16,133}

\definecolor{mycolor1}{rgb}{0.0, 0.18, 0.39}
\definecolor{mycolor2}{RGB}{87,108,67}
\definecolor{mycolor3}{RGB}{8,133,161}
\definecolor{mycolor4}{RGB}{80,91,161}
\definecolor{mycolor5}{RGB}{98,122,157}
\definecolor{mycolor6}{RGB}{255,163,67}
\definecolor{mycolor7}{RGB}{152,205,225}
\definecolor{mycolor8}{RGB}{242,204,48}
\definecolor{mycolor9}{rgb}{0,.5,0}
\definecolor{mycolor10}{rgb}{.59,.44,.09}
\definecolor{mycolor11}{RGB}{231,199,31} 
\definecolor{mycolor12}{RGB}{8,133,161} 
\definecolor{mycolor13}{RGB}{157,188,64} 
\definecolor{mycolor14}{RGB}{194,150,130} 
\definecolor{mycolor15}{RGB}{98,122,157} 
\definecolor{mycolor16}{RGB}{160,160,160} 
\definecolor{mycolor17}{RGB}{115,82,68} 
\definecolor{mycolor18}{RGB}{94,60,108} 
\definecolor{mycolor19}{RGB}{115,82,68} 
\definecolor{mycolor20}{RGB}{255,183,30} 
\definecolor{mycolor21}{RGB}{185,146,103} 

\usepackage{commath,amsfonts}
  \usepackage{amsmath}
   \usepackage{amssymb}
   \usepackage{bm}

\usepackage{footmisc}

\usepackage[justification=justified]{caption}
\DeclareCaptionFont{jkpl}{\fontsize{8}{11}\color{gray!70!black}\fontfamily{jkpl}\selectfont}
\usepackage{mathtools}
\usepackage{algorithmic}
\usepackage{pgfplots} 
\usetikzlibrary{arrows}
\pgfplotsset{compat=1.15}
\usepackage{graphicx}
\usepackage{xfrac}
\usepackage{colortbl}
\usepackage{cancel}
\usepgfplotslibrary{fillbetween}
\usepackage{float}
\usepackage{hyperref}
\usepackage{multirow}
\usepackage{xcolor}
\usepackage{mathrsfs}
\usepackage{bbm}
\usepackage{autobreak}
\allowdisplaybreaks
\usepackage{tikz}
\usetikzlibrary{arrows.meta}
\usepackage{amsbsy}
\usepackage{fixmath}
\usepackage{silence}
\usetikzlibrary{math}
\usepackage{graphicx}
\usepackage{pgfgantt}
\usepackage{pdflscape}

\usepackage{upgreek}
\usepackage{nicematrix}

\usepackage{leftindex}
\let\svtikzpicture\tikzpicture
\def\tikzpicture{\noindent\svtikzpicture}

\papertype{ArXiv Preprint}
\paperfield{Information Theory}
\title{{\normalfont\duinfamily\bfseries DETERMINISTIC IDENTIFICATION\\ FOR MC BINOMIAL CHANNEL}}

 \abbrevs{\fontfamily{jkpl}\selectfont \textcolor{gray!70!black}{ M. J. Salariseddigh, and C. Deppe are with the Institute for Communications Engineering, Technical University of Munich, Germany (e-mail: \{mjss,~christian.deppe\}@tum.de). V. Jamali is with the Resilient Communication Systems Group, Technical University of Darmstadt, Germany (e-mail: vahid.jamali@tu-darmstadt.edu). H. Boche is with the Chair of Theoretical Information Technology, Technical University of Munich, Germany (e-mail: boche@tum.de). R. Schober is with the Institute for Digital Communications, Friedrich-Alexander University Erlangen-N\"urnberg (FAU), Germany (e-mail: robert.schober@fau.de).}}

  \author{{\small\durmfamily MOHAMMAD JAVAD SALARISEDDIGH}\,, {\small\durmfamily VAHID JAMALI}\,, {\small\durmfamily HOLGER BOCHE}\,, {\small\durmfamily CHRISTIAN DEPPE}\,, and {\small\durmfamily ROBERT SCHOBER}}
\begin{document}
\begin{frontmatter}
\maketitle
\begin{abstract}
The Binomial channel serves as a fundamental model for molecular communication (MC) systems employing molecule-counting receivers. Here, deterministic identification (DI) is addressed for the discrete-time Binomial channels (DTBC), subject to an average and a peak constraint on the molecule release rate. We establish that the number of different messages that can be reliably identified for the DTBC scales as $2^{(n\log n)R}$, where $n$ and $R$ are the codeword length and coding rate, respectively. Lower and upper bounds on the DI capacity of the DTBC are developed.
\keywords{Deterministic identification, Binomial channel, molecular communications, and molecule-counting receivers}
\end{abstract}
\end{frontmatter}
\section{Introduction}
\fontfamily{jkpl}\selectfont
Molecular communication (MC) is a new communication strategy where information carriers are signaling molecules \cite{FYECG16}. Over the past decade, synthetic MC has been investigated in a number of different directions including channel modeling \cite{Jamali19}, modulation and detection design \cite{Kuscu19}, biological building blocks for transceiver design \cite{Soldner20}, and information-theoretical performance characterization \cite{Gohari16}. Furthermore, several proof-of-concept implementations for synthetic MC systems have been reported, see, e.g., \cite{Farsad17}.

Several applications of MC within the platform of future generation wireless networks (XG) \cite{6G_PST} give rise to event-triggered communication scenarios. Further discussions of the potentials of MC and post Shannon communication for the 6G can be found in \cite{Schwenteck23}. In such systems, Shannon's message transmission capacity, as studied early in \cite{S48}, may not be the appropriate performance metric, instead, the identification capacity is regarded to be a key quantitative measure. In particular, for object-finding or alarm-prompt scenarios, where the receiver determines the presence of an object or occurrence of an specific event by a reliable Yes\,/\,No answer, the so-called identification capacity is the applicable end-to-end performance measure \cite{AD89}.

While in Shannon's communication paradigm \cite{S48}, the sender encodes its message in a manner such that the receiver can perform a reliable estimation of \emph{all} messages, in the identification setting \cite{AD89}, the coding scheme is designed to accomplish a different objective, namely, to determine whether a \emph{particular} message was sent or not. The original problem of identification \cite{AD89} features a randomized encoding where each codeword is selected according to associated distribution. The size of the codebook for randomized identification (RI) problem grows double-exponentially in the codeword length $n$, i.e., $\sim 2^{ 2^{nR}}$ \cite{AD89}, where $R$ is the coding rate. Realization of RI codes entails high complexity and can be challenging for some applications; cf. \cite{Salariseddigh22}. Ahlswede and Dueck in \cite{AD89} were inspired to introduce the RI problem by the work of J\'aJ\'a \cite{J85} who considered a deterministic identification (DI) in communication complexity field \cite{AA20}. This problem may be also considered in the message transmission problem. In the DI problem, the codewords are selected via a deterministic function from the messages. DI may be preferred over randomized identification (RI) \cite{AD89} in complexity-constrained applications of MC systems, where the generation of random codewords is challenging. The DI for discrete memoryless and Gaussian channels is studied in \cite{Salariseddigh_IT}, where the codebook size scales exponentially and super-exponentially, i.e., $\sim 2^{nR}$ and $\sim 2^{(n\log n)R}$, respectively.
\pagecolor{gray!3!yellow!3}

The DI for discrete memoryless channels (DMCs) with average rate constraint is studied in \cite{Salariseddigh_ICC,Salariseddigh_IT} where size of the codebook similar to transmission problem of Shannon \cite{S48} grows exponentially in the codeword length \cite{AD89,Salariseddigh_ICC}. Furthermore, the DI for continuous alphabet channels including Gaussian channels with fast and slow fading and the memoryless discrete-time Poisson channel is addressed in \cite{Salariseddigh_ITW,Salariseddigh_IT,Salariseddigh_GC_IEEE,Salariseddigh_GC_arXiv} where a new observation regarding size of the codebook is established, namely, it scales super-exponentially with the codeword length $n$, i.e., $\sim 2^{(n\log n)R}$. A generalization of DI problem, called deterministic $K$-identification for the Gaussian channel with slow fading is studied in \cite{Salariseddigh22-3} where $K$-depending bounds on the DI capacity are established. The discrete time Poisson channel with inter-symbol interference (ISI) is studied in \cite{Salariseddigh22_2} where ISI-dependent bounds on the DI capacity are calculated.

In the context of MC, information can be encoded in the concentration (rate) of molecules released by transmitter and can be decoded based on the number of molecules reaching the receiver. Assuming that the release, propagation, and reception of different molecules are independent from each  other, the MC systems with molecule counting receivers are characterized by the Binomial channel. The transmission capacity of the Binomial channel is studied in \cite{Komninakis01,Wesel18}. In \cite{Farsad17-2} and \cite{Farsad20}, the Binomial channel law is used in modelling with imperfect particle intensity modulation and detection, is exploited for an MC channel. 
In the literature, the Binomial channel is often approximated by the Poisson channel when the number of released molecules, denoted by $N$, is large \cite[Sec.~IV]{Jamali19}, for which bounds on the DI capacity are studied in \cite{Salariseddigh22}. However, to the best of the authors' knowledge, the fundamental performance limits of DI for the original Binomial channel, which does not rely on very large $N$, has not been so far investigated in the literature.


\subsection{Contributions}
Here, we consider identification systems employing deterministic encoder and receivers that are interested to accomplish the identification task, namely, finding an object in a set of size $M$ where $M=2^{(n\log n)R}$. We assume that the communication over $n$ channel uses are independent of each other. We formulate the problem of DI over the discrete-time Binomial channels (DTBC) under average and peak constraints which account for the restricted molecule production\,/\,release rates at the transmitter. As our main objective, we investigate the fundamental performance limits of DI over the Binomial channel. In particular, this work makes the following contributions:
\vspace{.5mm}
\begin{itemize}[leftmargin=*]
    \item \textbf{\textcolor{mycolor12}{Codebook Scale}}: We establish that the codebook size of DI problem over the Binomial channels with average and peak  constraints on the molecule release rate for deterministic encoding scales super-exponentially in the codeword length ($\sim 2^{(n\log n)R}$). This result is in contrast with the scaling of the codebook size for conventional transmission (i.e., $2^{nR}$ \cite{S48}) and RI (i.e., $2^{2^{nR}}$ \cite{AD89}). The enlarged codebook size of the identification problem compared to the transmission problem may have interesting implications for MC system design.
    \item \textbf{\textcolor{mycolor12}{Capacity Bounds}}: We derive lower and upper bounds on the DI capacity of the DTBC, which are the main results of this work. Such bounds does not reflect the impact of rate constraints $P_{\,\text{ave}},P_{\,\text{max}}$ in the super-exponential scale, unless requiring them to be positive and finite values.
    \item \textbf{\textcolor{mycolor12}{Technical Novelty}}: To obtain the proposed lower bound, the existence of an appropriate sphere packing within the input space, for which the distance between the centers of the spheres does not fall below a certain value, is guaranteed. While the radius of the small spheres in the Gaussian case \cite{Salariseddigh_ITW} tends to zero, here the radius grows\footnote{\,\textcolor{gray!70!black}{In particular, we pack hyper spheres with radius $\sqrt{n\epsilon_n} \sim n^{\frac{1}{4}}$, inside a larger hyper sphere, which results in $\sim 2^{\frac{1}{4} n\log n}$ codewords.}} in the codeword length, $n$. Yet, we show that we can pack a super-exponential number of spheres within the larger cube. For derivation of the upper bound, we assume that for given sequences of codes with vanishing error probabilities, a certain minimum distance between the codewords is asserted, where this distance decreases as $n$ grows. In general, the derivation of upper bound for Binomial channel is more involved compared to that for the Gaussian \cite{Salariseddigh_ITW} and Poisson channels \cite{Salariseddigh_GC_IEEE} and entails employing of new analysis and inequalities. Here, proving the continuity of Binomial law requires dealing with binomial coefficients and factorial terms. We used inequalities on the ratio of two Gamma function depending on the distance of two codeword's symbols in every possible cases. While for the Gaussian channels with fading \cite{Salariseddigh_ITW}, the converse proof was based on establishing a minimum distance between Euclidean norm of each pair of codewords. Here, we consider a distance (absolute value norm) between symbols of two different codeword in the converse Lemma; cf. \eqref{Ineq.Conv_Distance}.
\end{itemize}

\subsection{Notations}
We use the following notations throughout this work:\\
Blackboard bold letters $\mathbbmss{K,X,Y,Z},\ldots$ are used for alphabet sets. Lower case letters $x,y,z,\ldots$ stand for constants and realization of random variables, and upper case letters $X,Y,Z,\ldots$ stand for random variables. $\lfloor x \rfloor$ stands for the greatest integer less than or equal to $x$. Lower case bold symbol $\fx$ and $\fy$ stand for row vectors of size $n$, that is, $\fx = (x_1, \dots, x_n)$ and $\fy = (y_1, \dots, y_n)$. The distribution of a random variable $X$ is specified by a probability mass function (pmf) $p_X(x)$ over a finite set $\X$. All logarithms and information quantities are for base $2$. The set of consecutive natural numbers from $1$ through $M$ is denoted by $[\![M]\!]$. We use symbol $\triangleq$ to specify a definition convention. The set of whole numbers is denoted by $\mathbb{N}_{0} \triangleq \{0,1,2,\ldots\}$. The set of real and negative numbers are denoted by $\mathbb{R}$ and $\mathbb{R}_{+}$, respectively. The gamma function for non-positive integer $x$ is denoted by $\Gamma(x)$ and is defined as $\Gamma (x) = (x-1) !$, where $(x-1)! \triangleq (x-1) \times \dots \times 1$. We use the small O notation, $f(n) = o(g(n))$, to indicate that $f(n)$ is dominated by $g(n)$ asymptotically, that is, $\lim_{n\to\infty} \frac{f(n)}{g(n)} = 0$. The big O notation, $f(n) = \mathcal{O}(g(n))$, is used to indicate that $|f(n)|$ is bounded above by $g(n)$ (up to constant factor) asymptotically, that is, $\limsup_{n\to\infty} \frac{|f(n)|}{g(n)} < \infty$. The the $\ell_1$-norm, $\ell_2$-norm and $\ell_{\infty}$-norm of vector $\fx$ are denoted by $\norm{\mathbf{x}}_1$, $\norm{\mathbf{x}}$, and $\norm{\mathbf{x}}_{\infty}$, respectively. Furthermore, we denote the $n$-dimensional hyper sphere of radius $r$ centered at $\fx_0$ with respect to the $\ell_2$-norm by $\S_{\fx_0}(n,r) = \{\fx\in\mathbb{R}_{+}^n : \norm{\fx-\fx_0} \leq r \}$. We use $\mathbf{0} = (0,\ldots,0)$ to represent coordination of the origin. An $n$-dimensional cube with center $(\frac{A}{2},\ldots,\frac{A}{2})$ and a corner at the origin, i.e., $\mathbf{0} = (0,\ldots,0)$, whose edges have length $A$ is denoted by $\Q_{\f0}(n,A) = \{\fx \in \mathbb{R}^n : 0 \leq x_t \leq A, \forall \, t\in[\![n]\!] \}$. We denote the considered Binomial channel by $\bB$.
\section{System Model and Preliminaries}
\label{Sec.SysModel}
In this section, we present the adopted system model and establish some preliminaries regarding DI coding.
\subsection{System Model}
We address an identification-focused communication setup, where the decoder's purpose is accomplishing the following task: Determining whether or not a target message
was sent by the transmitter; see Figure~\ref{Fig.Binomial_Channel}. To attain this objective, a coded communication between the transmitter and the receiver over $n$ channel uses of an MC channel is established.
We assume that for a given channel use, the transmitter generates $N = \lfloor  T_s X \rfloor$ molecules where $X$ is the rate for molecule generation and $T_s$ denotes the symbol duration. The generated molecules are released instantaneously into the channel at the beginning of the next symbol interval. Let $p$ denote the probability that a molecule released by the transmitter, is observed at the receiver, whose value depends on the parameters such as the diffusion coefficient of the molecules $D$, the distance between the transmitter and the receiver $d$, and the type of reception (e.g. absorbing or transparent receivers); see \cite{Jamali19} for further details. For instance, assuming molecule propagation via diffusion in an unbounded three-dimensional environment and under the approximation of uniform concentration  within the reception volume of a transparent receiver, $p$ at sampling time $\tau$ after the release of molecules is obtained as~\cite{Jamali19}
\begin{align} 
    p = \frac{V_{\rm rx}}{(4\pi D \tau)^{3/2}} e^{-\frac{d^2}{4D\tau}} \,,
\end{align}
where $V_{\rm rx}$ is the reception volume size. Assuming that the release, propagation, and reception of molecules are independent from each other,   the probability of observing $Y$ molecules at the receiver follows a Binomial distribution\footnote{\,\textcolor{gray!70!black}{The Binomial channel can be approximated by the Poisson channel (for large $N$ and small $Np$) and the Gaussian channel (for large $N$ and large $Np$); cf. \cite{Gohari16,Jamali19,Damrath20}. However, here, we study the  Binomial channel which does not rely on the assumption of large~$N$.}}, i.e.,
\begin{align}
    \text{Binom}\left( \lfloor  T_s X \rfloor,p \right) = \binom{\lfloor T_s X \rfloor}{Y} p^Y (1-p)^{\lfloor  T_s X \rfloor-Y} \,.
\end{align}

While in principle, MC channels are dispersive, the contribution of ISI can be made negligible if the symbol intervals are chosen sufficiently large such that the channel impulse response (CIR) fully decays to zero within one symbol interval. Alternatively, enzymes \cite{Adam14} and reactive information molecules, such as acid\,/\,base molecules \cite{Farsad16-2}, \cite{Jamali18}, may be used to speed up the decay of the CIR as a function of time, which would increase the accuracy of this assumption. Therefore, in such scenarios, we can assume that the MC channel between transmitter and receiver is characterized by a memoryless Binomial channel $\bB$, that is, the $n$ channel uses are independent. Hence, the transition probability law for $n$ channel uses is given by
\begin{align}
    \label{Eq.Binomial_Channel_Law}
    W^n(\fy|\fx) 
    = \prod_{t=1}^n \binom{\lfloor  T_s x_t \rfloor}{y_t} p^{y_t} (1-p)^{\lfloor  T_s x_t \rfloor-y_t} \,,\,
\end{align}
where $\fx = (x_1,\dots,x_{n})$ and $\fy = (y_1,\dots,y_{n})$ denote the transmitted codeword and the received signal, respectively.

The peak and average release rate constraints on the codewords $\fx = (x_t)\big|_{t=1}^n$ are
\begin{align}
    \label{Ineq.Const_X}
    0 \leq x_{t} \leq P_{\,\text{max}} \hspace{5mm} \text{and} \hspace{5mm} \frac{1}{n}\sum_{t=1}^{n} x_t \leq P_{\,\text{avg}} \,,\,
\end{align}
respectively, $\forall t\in[\![n]\!]$, where $P_{\,\text{max}} > 0$ and $P_{\,\text{avg}} > 0$ constrain the maximum value of molecule release rate per channel use and average molecule release rate over the entire $n$ channel uses in each codeword, respectively. 

While, unlike the Poisson and Gaussian approximations, the Binomial model does not require that the number of released molecules to be \textit{very} large \cite{Jamali19}, in practice, the number of released molecules cannot be too small either for reliable communication. This observation motivates us to adopt the approximation $\lfloor  T_s x_t \rfloor\approx T_s x_t$ in the remaining part, since the relative error $\frac{T_s x_t - \lfloor  T_s x_t \rfloor}{T_s x_t}\leq \frac{1}{T_s x_t}$ become sufficiently small for \textit{reasonably} large $T_s x_t$.
\subsection{DI Coding For The Binomial Channel}

The definition of a DI code for the Binomial channel $\bB$ is given below.
\begin{definition}[Binomial DI Code]
\label{Def.Binomial-Code}
An $(n,\allowbreak M(n,R),\allowbreak K(n,\allowbreak \kappa), \allowbreak e_1, \allowbreak e_2)$ DI code for a Binomial channel $\bB$ under average and peak rate constraints of $P_{\,\text{ave}}$ and $P_{\,\text{max}}$, and for integers $M(n,R)$ and $K(n,\kappa)$, where $n$ and $R$ are the codeword length and coding rate, respectively, is defined as a system $(\C,\mathscr{T})$, which consists of a codebook $\C \hspace{-.5mm} = \hspace{-.5mm} \big\{ \mathbf{c}_i \big\}_{i\in[\![M]\!]} \hspace{-.5mm} \subset \hspace{-.5mm} \mathbb{R}_{+}^n$, such that
\begin{align}
     0 \leq c_{i,t} \leq P_{\,\text{max}} \hspace{5mm} \text{and} \hspace{5mm} \frac{1}{n}\sum_{t=1}^{n} c_{i,t} \leq P_{\,\text{avg}} \,,\,
\end{align}
$\forall i\in[\![M]\!]$, $\forall t \in [\![n]\!]$ and a collection of decoding regions $\mathscr{T}=\{ \mathbbmss{T}_i \}_{i\in[\![M]\!]}$ where $\bigcup_{i=1}^{M(n,R)} \mathbbmss{T}_i \subset \mathbb{N}_0^n$.
Given a message $i \in [\![M]\!]$, the encoder sends $\mathbf{c}_i$, and the decoder's task is to address a binary hypothesis: Was a target message $j \in \mathbbmss{K}$ sent or not? There exist two types of errors that may happen:
\begin{description}
    \item [Type I]: Rejection of the \emph{actual} message, $i \in [\![M]\!]\,.$
    \item [Type II]: Acceptance of a \emph{wrong} message, $j \neq i\,.$
\end{description}
The associated error probabilities of the DI code $(\C,\sT)$ reads
\begin{align}
    \label{Eq.TypeIError}
    P_{e,1}(i) & = 1 -\sum_{\fy \in \mathbbmss{T}_{i}} W^n \left( \fy \, \big| \, \fc_i \right)
    \,,\
    \\
    P_{e,2}(i,j) & = \sum_{\fy \in \mathbbmss{T}_{j}} W^n \left( \fy \, \big| \, \fc_i \right)
    \,.\;
    \label{Eq.TypeIIError}
\end{align}
and satisfy the following bounds $P_{e,1}(i) \leq e_1 \,,\, \forall i \in [\![M]\!]$ and $P_{e,2}(i,j) \leq e_2 \,,\, \forall i \neq j$, and every $e_1, e_2>0$.

A DI rate $R>0$ is called achievable if for every $e_1, \allowbreak e_2>0$ and sufficiently large $n$, there exists an $(n,\allowbreak M(n\allowbreak,R),\allowbreak K(n,\allowbreak \kappa), \allowbreak e_1, \allowbreak e_2)$ DI code. The operational DI capacity of the Binomial channel $\bB$ is defined as the supremum of all achievable rates, and is denoted by $\mathbb{C}_{\rm DI}(\bB,M)$.
\begin{figure}[H]
\centering
\resizebox{.87\linewidth}{!}{\scalebox{1}{
\begin{tikzpicture}[scale=.5
][thick]

\draw[dashed] (0.07,.5) circle (3.5cm);

\draw (0,3.6) circle (.1cm);
\draw (3,1.5) circle (.1cm);
\draw (2,-.1) circle (.1cm);
\draw (-2,2) circle (.1cm);
\draw (-.1,.6) circle (.1cm);
\draw (1.4,-1.4) circle (.1cm);
\draw (-1.8,-.8) circle (.1cm);

\node at (0,3.05) {$\fc_2$};
\draw[dashed] (0,3.6) circle (0.2cm);
\draw [fill=cyan!40, fill opacity=0.7] (0,3.6) circle (.1cm);

\node at (3,1) {$\fc_3$};

\node at (2,-.5) {$\fc_4$};

\node at (-2,1.6) {$\fc_1$};
\node at (-.1,.2) {$\fc_5$};

\node at (1.4,-1.9) {$\fc_6$};
\draw[dashed] (1.4,-1.4) circle (.2cm);
\draw [fill=cyan!40, fill opacity=0.7] (1.4,-1.4) circle (.1cm);

\node at (-1.8,-1.2) {$\fc_7$};

\draw[->] (0,4.4) -- ++(0,1)  node [fill=white,inner sep=3pt](a){$$\text{\fontfamily{jkpl}\selectfont Input Space}$$};

\draw[->] (-2.3+15.07,4.4) -- ++(0,1)  node [fill=white,inner sep=3pt](a){$\text{\fontfamily{jkpl}\selectfont Output Space}$};

\draw[dashed] (-0.3+13.07,.4) circle (3.5cm);

\draw (-0.3+14.50,-.5) circle (1.2cm); 
\draw [fill=gray!20, fill opacity=0.7] (-0.3+14.50,-.5) circle (1.2cm);
\node at (-0.1+14.50,-.5) {$\mathbbmss{T}_{5}$};

\draw (-4.3+15.50,1.5) circle (1.2cm);    
\draw [fill=gray!20, fill opacity=0.7] (-4.3+15.50,1.5) circle (1.2cm);
\node at (-4.3+15.30,1.5) {$\mathbbmss{T}_{1}$};

\draw (-2.3+15.07,-1.4) circle (1.2cm);    
\draw [fill=gray!20, fill opacity=0.7] (-2.3+15.07,-1.4) circle (1.2cm);
\node at (-2.3+15.07,-1.6) {$\mathbbmss{T}_{6}$};

\draw (-3.3+14.50,-.4) circle (1.2cm);    
\draw [fill=gray!20, fill opacity=0.7] (-3.3+14.50,-.4) circle (1.2cm);
\node at (-3.3+14.30,-.6) {$\mathbbmss{T}_{7}$};

\draw (-2.3+15.07,2.4) circle (1.2cm);     
\draw [fill=gray!20, fill opacity=0.7] (-2.3+15.07,2.4) circle (1.2cm);
\node at (-2.3+15.07,2.6) {$\mathbbmss{T}_{2}$};

\draw (-0.3+14.50,1.4) circle (1.2cm);     
\draw [fill=gray!20, fill opacity=0.7] (-0.3+14.50,1.4) circle (1.2cm);
\node at (-0.3+14.70,1.4) {$\mathbbmss{T}_{3}$};

\draw (-2.4+15.07,.5) circle (1.2cm);    
\draw [fill=gray!20, fill opacity=0.7] (-2.4+15.07,.5) circle (1.2cm);
\node at (-2.4+15.07,.5) {$\mathbbmss{T}_{4}$};

\path (0.2,3.7) edge [-> , thick, mycolor9, bend left] node [sloped,midway,above,font=\small] {$\text{\fontfamily{jkpl}\selectfont Correct Identification}$}(13,3.1);
\path (0.2,3.6) edge [-> , thick, rot_8b, bend right] node [sloped,midway,below,font=\small] {\text{\fontfamily{jkpl}\selectfont Type I Error}}(11.3,2.2);
\path (1.6,-1.3) edge [-> , thick, rot_9b, bend right] node [sloped,midway,below,font=\small] {\text{\fontfamily{jkpl}\selectfont Type II Error}}(12.7,2.1);

\end{tikzpicture}
}}
\caption{Illustration of a deterministic identification setting. Assuming that the decision maker is the decoder $\mathbbmss{T}_2$, in the correct identification scenario, channel output is observed in the decoder $\mathbbmss{T}_2$ whose index coincide the sent message. Type I (miss-identification) error occurs if the channel output is detected in the complement of the decoder whose index is identical to the sent message and type II error (false identification) happens where the index of decoder for which channel output belongs to, differs from the sent message.
}
\label{Fig.DI-Code}
\end{figure}
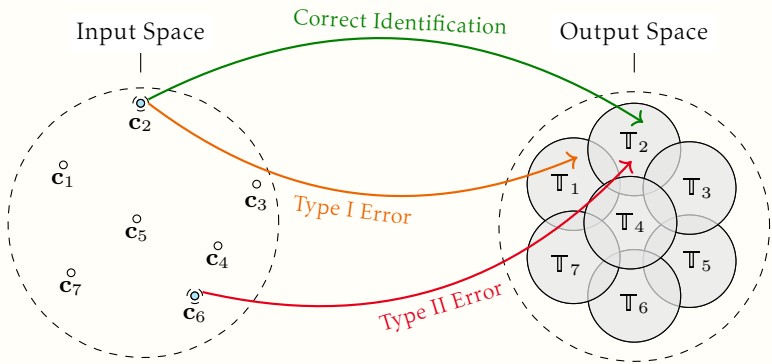
\end{definition}
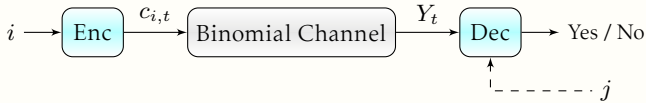
\begin{figure}[!t]
    \centering
	\tikzstyle{l} = [draw, -latex']
\tikzstyle{Block1} = [draw, top color = white, middle color = cyan!30, rectangle, rounded corners, minimum height=2em, minimum width=2.5em]
\tikzstyle{Block3} = [draw, top color = white, middle color = gray!20, rectangle, rounded corners, minimum height=2em, minimum width=1em]
\tikzstyle{Block4} = [draw, top color = white, middle color = orange_6b!30, rectangle, rounded corners, minimum height=2em, minimum width=2.5em]
\tikzstyle{Block5} = [draw, top color = white, middle color = cyan!30, rectangle, rounded corners, minimum height=2em, minimum width=2.5em]
\tikzstyle{input} = [coordinate]
\tikzstyle{sum} = [draw, circle,inner sep=0pt, minimum size=5mm,  thick]
\tikzstyle{conv} = [draw, circle,inner sep=0pt, minimum size=5mm,  thick]
\tikzstyle{arrow}=[draw,->]
\tikzstyle{small_node} = [draw, circle,inner sep=0pt, minimum size=.8mm,thick]
\begin{tikzpicture}[auto, node distance=2cm,>=latex']
\node[] (M) {$i$};
\node[Block1,right=.5cm of M] (enc) {\text{\fontfamily{jkpl}\selectfont Enc}};
\node[Block3, right=.84cm of enc] (channel) {$\text{\fontfamily{jkpl}\selectfont Binomial Channel}$};
\node[Block5, right=.84cm of channel] (dec) {\text{\fontfamily{jkpl}\selectfont Dec}};
\node[right=.5cm of dec] (Output) {$\text{\fontfamily{jkpl}\selectfont \small Yes\,/\,No}$};
\node[below=3mm of Output] (Target) {$j$};
%
\draw[->] (M) -- (enc);
\draw[->] (enc) -- node [above] {$\fontfamily{jkpl}\selectfont c_{i,t}$} (channel);
\draw[->] (channel) -- node [above] {$Y_t$} (dec);
\draw[->] (dec) -- (Output);
\draw[dashed,<-] (dec) |- (Target);

\end{tikzpicture}
	\vspace{-1mm}
	\caption{End-to-end transmission chain for DI communication in a generic molecular communication system modelled as a Binomial channel. The transmitter maps message $i$ onto a codeword $\fc_i = (c_{i,t})|_{t=1}^n$. The receiver is provided with an arbitrary message $j$, and given the channel output vector $\fY = (Y_t)|_{t=1}^n$, it asks whether $j$ is identical to $i$ or not.}   
    \label{Fig.Binomial_Channel}
    \vspace{-2mm}
\end{figure}
\section{DI Capacity of the Binomial Channel}
\label{Sec.Res}

In this section, we first present our main results, i.e., lower and upper bounds on the achievable identification rates for the Binomial channel. Subsequently, we provide the detailed proofs of these bounds.

\subsection{Main Results}
The DI capacity theorem for the Binomial channel $\bB$ is stated below.
\begin{theorem}
    \label{Th.DI-Capacity}
    Consider the Binomial channel $\bB$ subject to average and peak power constraints of the form $n^{-1} \sum_{t=1}^n c_{i,t} \leq P_{\,\text{ave}}$ and $0 \leq c_{i,t} \leq P_{\,\text{max}}$, respectively. Then, the DI capacity in the super-exponential scale, i.e., $M(n,R) = 2^{(n\log n)R}$, is bounded by
    \begin{align}
        \label{Ineq.LU}
        \frac{1}{4} \leq \mathbb{C}_{\rm DI}(\B,M) \leq \frac{3}{2} \,.\;
    \end{align}
\end{theorem}
\begin{proof}
The proof of Theorem~\ref{Th.DI-Capacity} consists of two parts, namely the achievability and the converse proofs, which are provided in Sections~\ref{Sec.Achiev} and \ref{Sec.Conv}, respectively.
\end{proof}
\subsection{Lower Bound (Achievability Proof)}
\label{Sec.Achiev}
The achievability proof consists of the following two main steps.
\begin{itemize}
    \item \textbf{\textcolor{mycolor12}{Step 1}}: First, we propose a codebook construction and  derive an analytical lower bound on the corresponding codebook size using inequalities for sphere packing density.
    \item \textbf{\textcolor{mycolor12}{Step 2}}: Then, to prove that this codebook leads to an achievable rate, we propose a decoder and show that the corresponding type I and type II error rates vanished as $n \to \infty$.
\end{itemize}
\subsubsection*{Codebook Construction}
Let $A = \min \left( P_{\,\text{max}},P_{\,\text{ave}} \right)$. In the following, we confine ourselves to codewords that meet the condition $0 \leq c_t \leq A$, $\forall \, t \in [\![n]\!]$. We argue that this condition ensures both the average and the peak power constraints in \eqref{Ineq.Const_X}. In particular,
\begin{itemize}
    \item \textbf{Case 1:} $P_{\,\text{max}} \leq P_{\,\text{ave}}$, then $A = P_{\,\text{max}}$ and the constraint $0 \leq c_t \leq A \, \forall \, t \in [\![n]\!]$ yields $\frac{1}{n} \sum_{t=1}^n c_t \leq A = P_{\,\text{max}}^2 \leq P_{\,\text{ave}}$, that is the average power constraint $\frac{1}{n} \sum_{t=1}^n c_t \leq P_{\,\text{ave}}$ is met. Furthermore, condition $0 \leq c_t \leq A \, \forall \, t \in [\![n]\!]$ implies $0 \leq c_t \leq P_{\,\text{max}} \, \forall \, t \in [\![n]\!]$, i.e., the peak power constraint is attained.
    \item \textbf{Case 2:} $P_{\,\text{max}} > P_{\,\text{ave}}$, then $A = P_{\,\text{ave}}$. Now by $0 \leq c_t \leq A\, \forall \, t \in [\![n]\!]$, we obtain $\frac{1}{n} \sum_{t=1}^n c_t \leq A = P_{\,\text{ave}}$, that is, the average power constraint is fulfilled. Furthermore, the condition $0 \leq c_t \leq A \, \forall \, t \in [\![n]\!]$ implies $0 \leq c_t \leq P_{\,\text{ave}} \leq P_{\,\text{max}} \, \forall \, t \in [\![n]\!]$, that is, the peak power constraint is accomplished.
\end{itemize}
Hence, in the following, we restrict our considerations to a hyper cube with edge length $A$. We use a packing arrangement of non-overlapping hyper spheres of radius $r_0 = \sqrt{n\epsilon_n}$ in a hyper cube with edge length $A$, where
\begin{align}
    \label{Eq:epsilon}
    \epsilon_n = \frac{a}{n^{\frac{1}{2}(1-b)}} \;,\,
\end{align}
and $a>0$ is a non-vanishing fixed constant and $0 < b < 1$ is an arbitrarily small constant\footnote{\,\textcolor{gray!70!black}{We note that our achievability proof is valid for any $b\in(0,1)$; however, arbitrarily small values of $b$ leads to the tightest lower bound and hence are of interest here.}}.

Let $\mathscr{S}$ denote a sphere packing, i.e., an arrangement of $L$ non-overlapping spheres $\S_{\fc_i}(n,r_0)$, $i\in [\![L]\!]$, that are packed inside the larger cube $\Q_{\f0}(n,A)$ with an edge length $A$, see Figure~\ref{Fig.Density}. As opposed to standard sphere packing coding techniques \cite{CHSN13}, the spheres are not necessarily entirely contained within the cube. That is, we only require that the centers of the spheres are inside $\Q_{\f0}(n,A)$ and are disjoint from each other and have a non-empty intersection with $\Q_{\f0}(n,A)$. The packing density $\Updelta_n(\mathscr{S})$ is defined as the ratio of the saturated packing volume to the cube volume $\text{Vol}\left(\Q_{\f0}(n,A)\right)$, i.e.,
\begin{align}
    \Updelta_n(\mathscr{S}) \triangleq \frac{\text{Vol}\left(\bigcup_{i=1}^{L}\S_{\fc_i}(n,r_0)\right)}{\text{Vol}\left(\Q_{\f0}(n,A)\right)} \,.\,
    \label{Eq.DensitySphereFast}
\end{align}
Sphere packing $\mathscr{S}$ is called \emph{saturated} if no spheres can be added to the arrangement without overlap.
\begin{figure}[t]
    \centering
	\scalebox{.85}{
\begin{tikzpicture}[scale=.58,rotate=0][thick]

\draw (-1.41,-1.41) circle (1cm);
\draw [fill=gray!20, fill opacity=0.35] (-1.41,-1.41) circle (1cm);
\node [fill=black, shape=circle, inner sep=.8pt] ($.$) at (-1.41,-1.41) {};
\draw (0,0) circle (1cm);
\draw [fill=gray!20, fill opacity=0.35] (0,0) circle (1cm);
\draw (1.41,1.41) circle (1cm);
\draw [fill=gray!20, fill opacity=0.35] (1.41,1.41) circle (1cm);
\node [fill=black, shape=circle, inner sep=.8pt] ($.$) at (1.41,1.41) {};
\draw (+.52,-1.93) circle (1cm);
\draw [fill=gray!20, fill opacity=0.35] (+.52,-1.93) circle (1cm);
\node [fill=black, shape=circle, inner sep=.8pt] ($.$) at (+.52,-1.93) {};
\draw (1.93,-.52) circle (1cm);
\draw [fill=gray!20, fill opacity=0.35] (1.93,-.52) circle (1cm);
\node [fill=black, shape=circle, inner sep=.8pt] ($.$) at (1.93,-.52) {};
\draw (-1.93,+.52) circle (1cm);
\draw [fill=gray!20, fill opacity=0.35] (-1.93,.52) circle (1cm);
\node [fill=black, shape=circle, inner sep=.8pt] ($.$) at (-1.93,.52) {};
\draw (-.52,1.93) circle (1cm);
\draw [fill=gray!20, fill opacity=0.35] (-.52,1.93) circle (1cm);
\node [fill=black, shape=circle, inner sep=.8pt] ($.$) at (-.52,1.93) {};

\draw (2.45,-2.45) circle (1cm);
\draw [inner color=mycolor16,outer color=gray!20, fill opacity=0.5] (2.45,-2.45) circle (1cm);
\node [fill=black, shape=circle, inner sep=.8pt] ($.$) at (2.45,-2.45) {};
\draw (2.81,2.81) circle (1cm);
\draw [inner color=mycolor16,outer color=gray!20, fill opacity=0.5] (2.81,2.81) circle (1cm);
\node [fill=black, shape=circle, inner sep=.8pt] ($.$) at (2.81,2.81) {};
\draw (-2.45,2.45) circle (1cm);
\draw [inner color=mycolor16,outer color=gray!20, fill opacity=0.5] (-2.45,2.45) circle (1cm);
\node [fill=black, shape=circle, inner sep=.8pt] ($.$) at (-2.45,2.45) {};
\draw (-2.81,-2.81) circle (1cm);
\draw [inner color=mycolor16,outer color=gray!20, fill opacity=0.5] (-2.81,-2.81) circle (1cm);
\node [fill=black, shape=circle, inner sep=.8pt] ($.$) at (-2.81,-2.81) {};

\foreach \s in {3}
{
\draw [thick] (-\s,-\s) -- (\s,-\s) -- (\s,\s) -- (-\s,\s) -- (-\s,-\s);
}

\node [fill=black, shape=circle, inner sep=.8pt] ($.$) at (0,0) {};

\draw [dashed] (0,0) -- (-3,0) node [left,font=\large] {$A/2$};
\draw [dashed] (-2.81,-2.81) -- (-3.71,-2.83) node [left,font=\large] {$\sqrt{n\epsilon_n}$};
\draw [dashed] (3,-3) -- (-3,3) node [above left,font=\large] {$A\sqrt{n}$};
\end{tikzpicture}}
	\vspace{3mm}
	\caption{Illustration of a saturated sphere packing inside a cube, where small spheres of radius $r_0 = \sqrt{n\epsilon_n}$ cover a larger cube. Dark gray colored spheres are not entirely contained within the larger cube, and yet they contribute to the packing arrangement. As we assign a codeword to each sphere center, the $1$-norm and arithmetic mean of a codeword are bounded by $A$ as required.}
	\vspace{-5mm}
	\label{Fig.Density}
\end{figure}
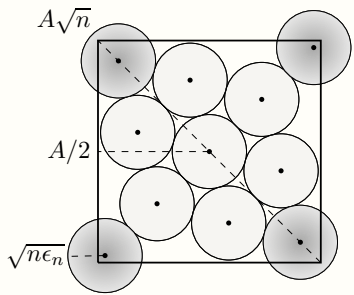
In particular, we use a packing argument that has a similar flavor as that observed in the Minkowski--Hlawka theorem for saturated packing \cite{CHSN13}.

Specifically, consider a saturated packing arrangement of 
\begin{align}
    \label{Eq.Union_Spheres}
    \bigcup_{i=1}^{M(n,R)} \S_{\fc_i}(n,\sqrt{n\epsilon_n})
\end{align}
spheres with radius $r_0=\sqrt{n\epsilon_n}$ embedded within cube $\Q_{\f0}(n,A)$. Then, for such an arrangement, we have the following lower \cite[Lem.~2.1]{C10} and upper bounds \cite[Eq.~45]{CHSN13} on the packing density 
\begin{align}
    \label{Ineq.Density}
    2^{-n} \leq \Updelta_n(\mathscr{S}) \leq 2^{-0.599n} \;.\,
\end{align}
In our subsequent analysis, we use the above lower bound which can be proved as follows: For the saturated packing arrangement given in \eqref{Eq.Union_Spheres}, there cannot be a point in the larger cube $\Q_{\f0}(n,A)$ with a distance of more than $2r_0$ from all sphere centers. Otherwise, a new sphere could be added which contradicts the assumption that the union of $M(n,R)$ spheres with radius $\sqrt{n\epsilon_n}$ is saturated. Now, if we double the radius of each sphere, the spheres with radius $2r_0$ cover thoroughly the entire volume of $\Q_{\f0}(n,A)$, that is, each point inside the hyper cube $\Q_{\f0}(n,A)$ belongs to at least one of the small spheres. In general, the volume of a hyper sphere of radius $r$ is given by \cite[Eq.~(16)]{CHSN13}
\begin{align}
    \text{Vol}\left(\S_{\fx}(n,r)\right) = \frac{\pi^{\frac{n}{2}}}{\Gamma(\frac{n}{2}+1)} \cdot r^{n} \,.\,
    \label{Eq.VolS}
\end{align}
Hence, if the radius of the small spheres is doubled, the volume of $\bigcup_{i=1}^{M(n,R)} \S_{\fc_i}(n,\sqrt{n\epsilon_n})$ is increased by $2^n$. Since the spheres with radius $2r_0$ cover $\Q_{\f0}(n,A)$, it follows that the original $r_0$-radius packing has a density of at least $2^{-n}$~\footnote{\,\textcolor{gray!70!black}{We note that the proposed proof of the lower bound in \eqref{Ineq.Density} is non-constructive in the sense that, while the existence of the respective saturated packing is proved, no systematic construction method is provided.}}.
We assign a codeword to the center $\fc_i$ of each small sphere. The codewords satisfy the input constraint as $0 \leq c_{i,t} \leq A$, $\forall t \in [\![n]\!]$, $\forall i\in [\![L]\!]$, which is equivalent to
\begin{align}
    \label{Ineq.Norm_Infinity}
    \norm{\fc_i}_{\infty} \leq A \;.\,
\end{align}
Since the volume of each sphere is equal to $\text{Vol}(\S_{\fc_1}(n,r_0))$ and the centers of all spheres lie inside the cube, the total number of spheres is bounded from below by
\begin{align}
    \label{Eq.M_Achiev}
    M & = \frac{\text{Vol}\left(\bigcup_{i=1}^{L}\S_{\fc_i}(n,r_0)\right)}{\text{Vol}(\S_{\fc_1}(n,r_0))}
    \nonumber\\&
    = \frac{\Updelta_n(\mathscr{S}) \cdot \text{Vol}\left(\Q_{\f0}(n,A)\right)}{\text{Vol}(\S_{\fc_1}(n,r_0))}
    \nonumber\\&
    \geq 2^{-n} \cdot \frac{A^n}{\text{Vol}(\S_{\fc_1}(n,r_0))}
    \,,\,
\end{align}
where the first inequality holds by (\ref{Eq.DensitySphereFast}) and the second inequality holds by (\ref{Ineq.Density}).

The above bound can be further simplified as follows
\begin{align}
    \log M & \geq \log \left( \frac{A^n}{\text{Vol}\left(\S_{\fc_1}(n,r_0)\right)} \right) - n
    \nonumber\\&
    \stackrel{(a)}{\geq} n\log\left( \frac{A}{\sqrt{\pi} r_0 } \right) + \log \left( \floor{\frac{n}{2}} ! \right) - n
    \nonumber\\&
    \stackrel{(b)}{=} n \log A - n \log r_0 + \floor{\frac{n}{2}} \log \left( \floor{\frac{n}{2}} \right) - \floor{\frac{n}{2}} \log e + o\left( \floor{\frac{n}{2}} \right) - n \;,\,
\end{align}
where $(a)$ exploits (\ref{Eq.VolS}) and
\begin{align}
    \Gamma\left( \frac{n}{2} + 1 \right) & \stackrel{(a)}{=} \frac{n}{2} \Gamma \left( \frac{n}{2} \right)
    \nonumber\\
    & \stackrel{(b)}{\geq} \floor{\frac{n}{2}} \Gamma \left( \floor{\frac{n}{2}} \right)
    \nonumber\\
    & \triangleq \floor{\frac{n}{2}} !
    \,,    
\end{align}
where $(a)$ holds by the recurrence relation of the Gamma function for a positive real argument and $(b)$ follows from $\floor{\frac{n}{2}} \leq \frac{n}{2}$ for positive integer $n$ and the monotonicity of the Gamma function for $n \geq 4 \equiv \floor{\frac{n}{2}} \in [z_1,\infty)$ where $z_1 \approx 1.46$ is the first root of the Digamma function; and $(b)$ follows from Stirling's approximation, that is, $\log n! = n \log n - n \log e + o(n)$ for integer $n$, \cite[P.~52]{F66}. Now, for $r_0 = \sqrt{n\epsilon_n} = \sqrt{a}n^{\frac{1+b}{4}}$, we obtain
\begin{align}
    & \log M
    \nonumber\\&
    \geq n \log \frac{A}{\sqrt{a}} - \frac{1}{4}(1+b) \, n \log n + \floor{\frac{n}{2}} \log \left( \floor{\frac{n}{2}} \right) - \floor{\frac{n}{2}} \log e + o\left( \floor{\frac{n}{2}} \right) - n
    \nonumber\\&
    \stackrel{(a)}{>} n \log \frac{A}{\sqrt{a}} - \frac{1}{4}(1+b) \, n \log n + \left( \frac{n}{2} - 1 \right) \log \left( \frac{n}{2} - 1 \right) - \floor{\frac{n}{2}} \log e + o\left( \frac{n}{2} - 1 \right) - n
    \nonumber\\&
    \stackrel{(b)}{\geq} n \log \frac{A}{\sqrt{a}} - \frac{1}{4}(1+b) \, n \log n + \left( \frac{n}{2} - 1 \right) \left( \log \left( \frac{n}{2} \right) - 1 \right) - \frac{n}{2} \log e + o\left( \frac{n}{2} - 1 \right) - n
    \nonumber\\&
    \stackrel{(c)}{=} n \log \frac{A}{\sqrt{a}} - \frac{1}{4}(1+b) \, n \log n + \frac{1}{2} n\log n - 2n - \log n  - \frac{n}{2} \log e + o\left( \frac{n}{2} - 1 \right)
    \nonumber\\&
    = \left( \frac{1-b}{4} \right) \, n \log n + n \left( \log \frac{A}{e\sqrt{a}} \right) - 2n - \log n  - \frac{n}{2} \log e + o\left( \frac{n}{2} - 1 \right)
    \;,\,
    \label{Eq.Log_M}
\end{align}
where $(a)$ holds by $\floor{\frac{n}{2}} > \frac{n}{2} - 1$ for integer $n$, $(b)$ holds since $\log(t-1) \geq \log t - 1$ for $t \geq 2$ and $\floor{\frac{n}{2}} \leq \frac{n}{2}$ for integer $n$, and $(c)$ follows since base of logarithm is 2. Observe that the dominant term \eqref{Eq.Log_M} is of order $n \log n$. Hence, for obtaining a finite value for the lower bound of the rate, $R$, \eqref{Eq.Log_M} induces the scaling law of $M$ to be $2^{(n\log n)R}$. Therefore, we obtain
\begin{align}
    R & \geq \frac{1}{n\log n} \left[ \left( \frac{1-b}{4} \right) n\log n + n \log \left( \frac{A}{e\sqrt{a}} \right) - 2n - \log n  - \frac{n}{2} \log e + o\left( \frac{n}{2} - 1 \right) \right] \;,\,
\end{align}
which tends to $\frac{1}{4}$ when $n \to \infty$ and $b\rightarrow 0$.
\subsubsection*{Encoding}
Given a message $i \in [\![M]\!]$, transmit $\fx = \fc_i$.
\subsubsection*{Decoding}
Fix $e_1,e_2 > 0$ and let $\zeta_0, \zeta_1 > 0$ be arbitrarily small constants. Before we proceed, for the sake of brevity of analysis, we introduce the following conventions. Let:
\begin{itemize}
    \item $Y_t(i) \sim \text{Binom($c_{i,t}T_s$,p)}$ denote the channel output at time $t$ \emph{conditioned} that $\fx=\fc_i$ was sent.
    \item $\fY(i)= (Y_1(i),\ldots, Y_{n}(i))$
\end{itemize}
Let
\begin{align}
    \label{Eq.delta_n}
    \delta_n & = \frac{A}{n^{\frac{1}{2}(1-b)}}
    \,,
\end{align}
where $0 < b < 1$  is an arbitrarily small constant.
To identify whether a message $j \in [\![M]\!]$ was sent, the decoder checks whether the channel output $\fy$ belongs to the following decoding set,
\begin{align}
    \mathbbmss{T}_j & = \left\{\fy \in \mathbb{N}_0^n \,:\; \left| T(\fy,\fc_j) \right| \leq \delta_n \right\}
    \,.
\end{align}
where
\begin{align}
    T(\fy,\fc_j) = \frac{1}{n} \sum_{t=1}^n \left( y_t - pT_sc_{j,t} \right)^2 - (1-p)y_t \,,
\end{align}
is referred to as the \emph{decoding metric} evaluated for observation vector $\fy$ and codeword $\fc_j$.
\subsubsection*{Error Analysis}
Consider the type I error, i.e., when the transmitter sends $\fc_i$, yet $\fY \notin \mathbbmss{T}_i$. For every $i \in [\![M]\!]$, the type I error probability is bounded by
\begin{align}
    \label{Eq.TypeI}
    P_{e,1}(i) = \Pr\left( \left| T(\fY(i),\fc_j) \right| > \delta_n \right) \,,
\end{align}
In order to bound $P_{e,1}(i)$, we apply the Chebyshev's inequality, namely
\begin{align}
    \label{Ineq.TypeI_Cheb}
    \Pr\left( \left| T(\fY(i),\fc_j) - \mathbb{E}\left[ T(\fY(i),\fc_j) \right] \right| > \delta_n \right) \leq \frac{\text{Var}\left[ T(\fY(i),\fc_j) \right]}{\delta_n^2} \,.
\end{align}
First, we calculate the expectation of the decoding metric as follows
\begin{align}
    \label{Ineq.TypeI_Expectation}
    & \mathbb{E}\left[ \frac{1}{n} \sum_{t=1}^n \left( Y_t(i) - pT_sc_{i,t} \right)^2 - (1-p)Y_t(i) \right]
    \nonumber\\
    & = \frac{1}{n} \sum_{t=1}^n \mathbb{E}\left[ \left( Y_t(i) - pT_sc_{i,t} \right)^2 - (1-p)Y_t(i) \right]
    \nonumber\\&
    = \frac{1}{n} \sum_{t=1}^n \mathbb{E}\left[ \left( Y_t(i) - pT_sc_{i,t} \right)^2 \right] - \mathbb{E}\left[ (1-p)Y_t(i) \right]
    \nonumber\\&
    = \frac{1}{n} \sum_{t=1}^n \text{Var}\left[ Y_t(i) \right] - \mathbb{E}\left[ Y_t(i) \right] (1-p)
    \nonumber\\&
    = \frac{1}{n} \sum_{t=1}^n pT_sc_{i,t}(1-p) - pT_sc_{i,t}(1-p)
    \nonumber\\&
    = 0
    \,.
\end{align}
Second, since the channel is memoryless, we can derive the variance of decoding metric as follows
\begin{align}
    \label{Ineq.TypeI_Variance-1}
    & \text{Var}\left[ \frac{1}{n} \sum_{t=1}^n \left( Y_t(i) - pT_sc_{i,t} \right)^2 - (1-p)Y_t(i) \right]
    \nonumber\\
    & = \frac{1}{n^2} \text{Var}\left[ \sum_{t=1}^n \left( Y_t(i) - pT_sc_{i,t} \right)^2 - (1-p)Y_t(i) \right]
    \nonumber\\
    & \stackrel{(a)}{=} \frac{1}{n^2} \sum_{t=1}^n \text{Var}\left[ \left( Y_t(i) - pT_sc_{i,t} \right)^2 - (1-p)Y_t(i) \right]
    \nonumber\\&
    = \frac{1}{n^2} \sum_{t=1}^n \text{Var}\left[ Y_t^2(i) - \left( 2pT_sc_{i,t} + 1 - p \right) Y_t(i) \right] \,,
\end{align}
where $(a)$ holds since the channel is memoryless. Now we apply the identity $\text{Var}\big[ aX - bY \big] = a^2 \text{Var}[X] + b^2 \text{Var}[Y] - 2ab \text{Cov}[X,Y]$ to summand in \eqref{Ineq.TypeI_Variance-1}, namely,
\begin{align}
    \label{Ineq.TypeI_Variance-2}
    & \text{Var}\left[ Y_t^2(i) - \left( 2pT_sc_{i,t} + 1 - p \right) Y_t(i) \right]
    \nonumber\\&
    = \text{Var}\big[ Y_t^2(i) \big] + \left( 2pT_sc_{i,t} + 1 - p \right)^2 \text{Var}\left[ Y_t(i) \right] - \left( 2pT_sc_{i,t} + 1 - p \right) \text{Cov}\left[ Y_t^2(i),Y_t(i) \right] \,.
\end{align}
Now, in order to bound \eqref{Ineq.TypeI_Variance-1}, we use the following three knowledge. First, observe that for a Binomial variable $Y_t(i) \sim \text{Binom($c_{i,t}T_s$,p)}$, we have $\text{Var}[Y_t^2(i)] \leq \mathbb{E}[Y_t(i)^4]$. Second, the identity $\text{Cov}[X,Y] \leq \sqrt{\text{Var}[X] \cdot \text{Var}[Y]}$, provide upper bound for covariance of variables with finite variances. Third, triangle inequality provide upper bound for summation of terms with different sign, i.e., $a - b \leq |a - b| \leq |a| + |b|$. Therefore, 
\begin{align}
    \label{Ineq.TypeI_Variance-3}
    & \frac{1}{n^2} \sum_{t=1}^n \text{Var}\left[ Y_t^2(i) - \left( 2pT_sc_{i,t} + 1 - p \right) Y_t(i) \right]
    \nonumber\\
    & \leq \frac{\splitfrac{\sum_{t=1}^n \mathbb{E}\left[ Y_t^4(i) \right] + \left( 2pT_sc_{i,t} + 1 - p \right)^2 pT_sc_{i,t}}{+ \left( 2pT_sc_{i,t} + 1 - p \right) \sqrt{\mathbb{E}\left[ Y_t^4(i) \right] \cdot pT_sc_{i,t}(1-p)}}}{n^2}
    \nonumber\\&
    \stackrel{(a)}{=} \frac{\splitfrac{n\Big[ (c_{i,t}T_s)^4 \exp\left( 8 / pT_s c_{i,t} \right) + \left( 2AT_s + 1 \right)^2 AT_s}{ + \left( 2AT_s + 1 \right) \sqrt{(pT_s c_{i,t})^4 \exp\left( 8 / pT_s c_{i,t} \right) AT_s} \Big]}}{n^2}
    \nonumber\\&
    \leq \frac{\splitfrac{A^4T_s^4 \exp\left( 8 / pT_s c_{i,t} \right) + \left( 2AT_s + 1 \right)^2 AT_s}{+ \left( 2AT_s + 1 \right) A^2T_s^2 \sqrt{ \exp\left( 8 / pT_sc_{i,t} \right) AT_s}}}{n^2}
    \,,
\end{align}
where $(a)$ holds since $0 < p < 1$; and for a Binomial variable $Y_t(i) \sim \text{Binom($c_{i,t}T_s$,p)}$, the non-centered moments are upper bounded as follows
\begin{align}
    \mathbb{E}\big[ Y_t^k(i) \big] \leq \mathbb{E}^k\big[ Y_t(i) \big] \cdot \exp\Big( k^2 / 2\mathbb{E}\big[ Y_t(i) \big] \Big) \,.
\end{align}
Therefore, exploiting \eqref{Ineq.TypeI_Cheb}, \eqref{Ineq.TypeI_Expectation} and \eqref{Ineq.TypeI_Variance-3}, we can establish the following upper bound on the type I error probability given in \eqref{Eq.TypeI}, namely
\begin{align}
    & P_{e,1}(i)
    \nonumber\\
    & = \Pr\left( \left| T(\fY(i),\fc_j) \right| > \delta_n \right)
    \nonumber\\&
    = \frac{ A^4T_s^4 \exp\left( 8 / pT_s c_{i,t} \right) + \left( 2AT_s + 1 \right)^2 AT_s + \left( 2AT_s + 1 \right) A^2T_s^2 \sqrt{ \exp\left( 8 / pT_sc_{i,t} \right) AT_s} }{n\delta_n^2}
    \nonumber\\&
    \stackrel{(a)}{=} \frac{ A^4T_s^4 \exp\left( 8 / pT_s c_{i,t} \right) + \left( 2AT_s + 1 \right)^2 AT_s + \left( 2AT_s + 1 \right) A^2T_s^2 \sqrt{ \exp\left( 8 / pT_sc_{i,t} \right) AT_s} }{n^b}
    \nonumber\\&
    \leq e_1 \,,
\end{align}
where $(a)$ follows from \eqref{Eq.delta_n}.

Next, we address the type II error, i.e., when $\fY \in \mathbbmss{T}_j$ while the transmitter sent $\fc_i$. Then, for every $i,j \in[\![M]\!]$, where $i \neq j$, the type II error probability is given by
\begin{align}
    P_{e,2}(i,j) = \Pr \left( \left| T(\fY(i);\fc_j) \right| \leq \delta_n \right)
    \;.\,
    \label{Eq.TypeII-1}
\end{align}
where
\begin{align}
    \label{Eq.Decoding_Metric_TypeII-1}
    T(\fY(i);\fc_j) = \frac{1}{n} \sum_{t=1}^n \left( Y_t(i) - pT_sc_{j,t} \right)^2 - (1-p)Y_t(i) \,.
\end{align}
Observe that \eqref{Eq.Decoding_Metric_TypeII-1} can be expressed as follows
\begin{align}
    \label{Eq.Decoding_Metric_TypeII-2}
    T(\fY(i);\fc_j) = \frac{1}{n} \sum_{t=1}^n \left( Y_t(i) - pT_sc_{i,t} +( c_{i,t} - c_{j,t} ) pT_s \right)^2 - (1-p)Y_t(i) \,.
\end{align}
Observe that the sum in \eqref{Eq.Decoding_Metric_TypeII-2} can be expressed as
\begin{align}
    \sum_{t=1}^n \left( Y_t(i) - pT_sc_{i,t} + ( c_{i,t} - c_{j,t} ) pT_s \right)^2 & = \sum_{t=1}^n \left( Y_t(i) - pT_sc_{i,t} \right)^2 + \sum_{t=1}^n \left( ( c_{i,t} - c_{j,t} ) pT_s \right)^2
    \nonumber\\
    & + 2 \sum_{t=1}^n \left( Y_t(i) - pT_sc_{i,t} \right) \left( ( c_{i,t} - c_{j,t} )  pT_s \right) \,.
    \label{Eq.TypeII-3}
\end{align}
Then, define the following events
\begin{align}
    \label{Eq.Event_E_0}
    \E_0 & = \left\{ \fY \in \mathbb{N}_0^n \,:\; \left| \sum_{t=1}^n \left( Y_t(i) - pT_sc_{i,t} \right) \left( ( c_{i,t} - c_{j,t} )  pT_s \right) \right| > n\delta_n \right\} \,,
    \\
    \label{Eq.Event_E_1}
    \E_1 & = \left\{ \fY \,:\; \sum_{t=1}^n \left( Y_t(i) - pT_sc_{i,t} \right)^2 + \sum_{t=1}^n \left( ( c_{i,t} - c_{j,t} ) pT_s \right)^2 - (1-p)Y_t(i) \leq 2n\delta_n \right\} \,,
    \\
    \label{Eq.Event_E_i_j}
    \E_{i,j} & = \left\{ \fY \in \mathbb{N}_0^n \,:\; \left| \sum_{t=1}^n \left( Y_t(i) - pT_sc_{i,t} +( c_{i,t} - c_{j,t} ) pT_s \right)^2 - (1-p)Y_t(i) \leq  n\delta_n \right| \right\}
    \\
    \label{Eq.Event_E_i_j-2}
    \E_{i,j}' & = \left\{ \fY \in \mathbb{N}_0^n \,:\; \sum_{t=1}^n \left( Y_t(i) - pT_sc_{i,t} +( c_{i,t} - c_{j,t} ) pT_s \right)^2 - (1-p)Y_t(i) \leq  n\delta_n \right\}\;.\,
\end{align}
Then,
\begin{align}
    P_{e,2}(i,j) & = \Pr\left( \E_{i,j} \right)
    \nonumber\\&
    = \Pr \left( \left| \sum_{t=1}^n \left( Y_t(i) - pT_sc_{i,t} +( c_{i,t} - c_{j,t} ) pT_s \right)^2 - (1-p)Y_t(i) \right| \leq n\delta_n \right)
    \nonumber\\&
    \stackrel{(a)}{\leq} \Pr \left( \left| \sum_{t=1}^n \left( Y_t(i) - pT_sc_{i,t} +( c_{i,t} - c_{j,t} ) pT_s \right)^2 \right| - \left| \sum_{t=1}^n (1-p)Y_t(i) \right| \leq n\delta_n \right)
    \nonumber\\&
    \stackrel{(b)}{\leq} \Pr \left( \sum_{t=1}^n \left( Y_t(i) - pT_sc_{i,t} +( c_{i,t} - c_{j,t} ) pT_s \right)^2 - \sum_{t=1}^n (1-p)Y_t(i) \leq n\delta_n \right)
    \nonumber\\&
    = \Pr\big( \E_{i,j}' \big) \,,
    \label{Eq.TypeII-2}
\end{align}
where $(a)$ exploits the reverse triangle inequality, $|\alpha| - |\beta| \leq |\alpha - \beta|$, and $(b)$ holds since $\alpha,\beta \geq 0$.

Now, we apply the law of total probability to event $\mathcal{E}_{i,j}'$ over $\E_0$ and its complement $\E_0^c$, and obtain the following upper bound on the type II error probability,
\begin{align}
    \label{Eq.TypeIIError-E_0+E_1}
    P_{e,2}(i,j) & \leq \Pr\big( \E_{i,j}' \big)
    \nonumber\\&
    = \Pr\big( \E_{i,j}' \cap \E_0 \big) + \Pr\big( \mathcal{E}_{i,j}' \cap \E_0^c \big)
    \nonumber\\&
    \stackrel{(a)}{\leq} \Pr\big( \E_0 \big) + \Pr\big( \mathcal{E}_{i,j}' \cap \E_0^c \big)
    \nonumber\\&
    \stackrel{(b)}{=} \Pr\big( \E_0 \big) + \Pr\big( \E_1 \big) \,,
\end{align}
where $(a)$ follows from $\E_{i,j}' \cap \E_0 \subset \E_0$ and $(b)$ holds since the event $\mathcal{E}_{i,j}' \cap \E_0^c$ yields the event $\E_1$, with the following argument. Observe that,
\begin{align}
    & \Pr\big( \mathcal{E}_{i,j}' \cap \E_0^c \big) 
    \nonumber\\
    & = \Pr\left( \sum_{t=1}^n \left( Y_t(i) - pT_sc_{i,t} \right)^2 + \sum_{t=1}^n \left( ( c_{i,t} - c_{j,t} ) pT_s \right)^2 - (1-p)Y_t(i) \leq n\delta_n - (-n\delta_n) \right)
    \nonumber\\&
    \stackrel{(a)}{=} \Pr\left( \sum_{t=1}^n \left( Y_t(i) - pT_sc_{i,t} \right)^2 + \sum_{t=1}^n \left( ( c_{i,t} - c_{j,t} ) pT_s \right)^2 - (1-p)Y_t(i) \leq 2n\delta_n \right)
    \nonumber\\&
    \stackrel{(b)}{=} \Pr\big( \E_1 \big) \,,
\end{align}
where $(a)$ holds since given the complementary event $\E_0^c$, we obtain
\begin{align}
    \sum_{t=1}^n \left( Y_t(i) - pT_sc_{i,t} \right) \left( ( c_{i,t} - c_{j,t} )  pT_s \right) \geq - n\delta_n \;,
\end{align}
, and $(b)$ follows from \eqref{Eq.Event_E_1}.

Now, we proceed to bound $\Pr(\E_0)$. By Chebyshev's inequality, we can establish the following upper bound on $\Pr(\E_0)$ as follows
\begin{align}
    \Pr(\E_0) & = \Pr \left( \left| \sum_{t=1}^n \left( Y_t(i) - pT_sc_{i,t} \right) \left( ( c_{i,t} - c_{j,t} )  pT_s \right) \right| > n\delta_n \right)
    \nonumber\\&
    \leq \frac{\text{Var}\left[ \sum_{t=1}^n \left( Y_t(i) - pT_sc_{i,t} \right) \left( ( c_{i,t} - c_{j,t} )  pT_s \right) \right]}{(n\delta_n)^2}
    \nonumber\\&
    = \frac{\sum_{t=1}^n \text{Var}\left[ \left( Y_t(i) - pT_sc_{i,t} \right) \left( ( c_{i,t} - c_{j,t} )  pT_s \right) \right]}{(n\delta_n)^2}
    \nonumber\\&
    = \frac{\sum_{t=1}^n  \left( ( c_{i,t} - c_{j,t} )  pT_s \right)^2 \text{Var}\left[ \left( Y_t(i) - pT_sc_{i,t} \right) \right]}{(n\delta_n)^2}
    \nonumber\\&
    = \frac{T_s^2p^2 \sum_{t=1}^n \left( c_{i,t} - c_{j,t} \right)^2 \text{Var}\left[ Y_t(i) \right]}{(n\delta_n)^2}
    \nonumber\\&
    \leq \frac{T_s^2p^2 \sum_{t=1}^n \left( c_{i,t} - c_{j,t} \right)^2 \cdot ApT_s(1-p)}{(n\delta_n)^2}
    \nonumber\\&
    = \frac{AT_s^3p^3(1-p) \norm{\fc_i - \fc_j}^2}{(n\delta_n)^2} \,,
    \label{InEq.Event_E_0_Prob}
\end{align}
Observe that
\begin{align}
    \norm{\fc_i - \fc_j}^2 & \stackrel{(a)}{\leq} \left( \norm{\fc_i} + \norm{\fc_j} \right)^2
    \nonumber\\
    & \stackrel{(b)}{\leq} \left( \sqrt{n} \norm{\fc_i}_{\infty} + \sqrt{n} \norm{\fc_j}_{\infty} \right)^2
    \nonumber\\
    & \stackrel{(c)}{\leq} \left( \sqrt{n} A + \sqrt{n} A \right)^2
    \nonumber\\
    & = 4nA^2 \,,
\end{align}
where $(a)$ holds by the triangle inequality, $(b)$ follows since $\norm{.} \leq \sqrt{n} \norm{.}_{\infty}$ and $(c)$ is valid by \eqref{Ineq.Norm_Infinity}.
Hence,
\begin{align}
    \Pr(\E_0) & \leq \frac{T_s^3p^3 A(1-p) \norm{\fc_i - \fc_j}^2}{(n\delta_n)^2}
    \nonumber\\&
    \leq \frac{4A^2T_s^3p^3 A(1-p) n}{n^2\delta_n^2}
    \nonumber\\&
    \leq \frac{4A^3T_s^3p^3 (1-p)}{n\delta_n^2}
    \nonumber\\&
    = \frac{4A^3T_s^3p^3 (1-p)}{n^b}
    \nonumber\\&
    \triangleq \zeta_0 \,,
    \label{InEq.Event_E_0_Prob-2}
\end{align}
We now proceed with bounding $\Pr\left(\E_1 \right)$ as follows. Based on the codebook construction, each pair of codeword are distanced by at least $r_0 = \sqrt{n\epsilon_n}$, hence,
\begin{align}
    \norm{pT_s(\fc_i - \fc_j)}^2 & \geq T_s^2p^2 n\epsilon_n
    \nonumber\\&
    = 3n\delta_n \;,
\end{align}
where the equality holds by \eqref{Eq.delta_n}.
Thus, we can establish the following upper bound for event $\E_1$:
\begin{align}
    & \Pr(\E_1) 
    \nonumber\\
    & \leq \left( \sum_{t=1}^n \left( Y_t(i) - pT_sc_{i,t} \right)^2 - (1-p)Y_t(i) \leq 2n\delta_n - \norm{pT_s (\fc_i - \fc_j)}^2 \right)
    \nonumber\\&
    \leq \left( \sum_{t=1}^n \left( Y_t(i) - pT_sc_{i,t} \right)^2 - (1-p)Y_t(i) \leq 2n\delta_n - 3n\delta_n \right)
    \nonumber\\&
    = \left( \sum_{t=1}^n \left( Y_t(i) - pT_sc_{i,t} \right)^2 - (1-p)Y_t(i) \leq -n\delta_n \right)
    \nonumber\\&
    \stackrel{(a)}{\leq} \frac{\text{Var}\left[ \sum_{t=1}^n \left( Y_t(i) - pT_sc_{i,t} \right)^2 - (1-p)Y_t(i) \right]}{n^2\delta_n^2}
    \nonumber\\&
    \stackrel{(b)}{=} \frac{ A^4T_s^4 \exp\left( 8 / pT_s c_{i,t} \right) + \left( 2AT_s + 1 \right)^2 AT_s + \left( 2AT_s + 1 \right) A^2T_s^2 \sqrt{ \exp\left( 8 / pT_sc_{i,t} \right) AT_s} }{n^b}
    \nonumber\\&
    \triangleq \zeta_1 \,,
\end{align}
where $(a)$ follows from applying Chebyshev's inequality, $(b)$ holds by similar line of arguments as we made in type I error probability analysis, see \eqref{Ineq.TypeI_Cheb} and the derivations afterward.

Therefore, recalling \eqref{Eq.TypeIIError-E_0+E_1}, we obtain
\begin{align}
    P_{e,2}(i,j) \leq \Pr(\E_0) + \Pr(\E_1) \leq \zeta_0 + \zeta_1 \leq e_2 \,,\,
\end{align}
hence, $P_{e,2}(i,j) \leq e_2$ holds for sufficiently large $n$ and arbitrarily small $e_2 > 0$. We have thus shown that for every $e_1,e_2>0$ and sufficiently large $n$, there exists an $(n, M(n,R), K(n,\kappa), \allowbreak e_1, e_2)$ code.
\subsection{Upper Bound (Converse Proof)}
\label{Sec.Conv}
The converse proof consists of the following two main steps.
\begin{itemize}
    \item \textbf{\textcolor{mycolor12}{Step 1}}: We show in Lemma~\ref{Lem.Converse} that for any achievable rate (for which the type I and type II error probabilities vanish as $n\to\infty$), the distance between every pair of codeword should be at least larger than a threshold.
    \item \textbf{\textcolor{mycolor12}{Step~2}}: Employing Lemma~\ref{Lem.Converse}, we derive an upper bound on the codebook size of achievable DI codes.
\end{itemize}
We start with the following lemma which establish a lower bound on the letter-wise ratio for every pair of codewords.
\begin{lemma}
\label{Lem.Converse}
Suppose that $R$ is an achievable rate for the Binomial channel $\bB$. Consider a sequence of $(n, M(n,R), K(n,\kappa), \allowbreak e_1^{(n)}, \allowbreak e_2^{(n)})$ codes $(\C^{(n)},\T^{(n)})$ such that $e_1^{(n)}$ and $e_2^{(n)}$ tend to zero as $n\rightarrow\infty$. Then, given a sufficiently large $n$, the codebook $\C^{(n)}$ satisfies the following property. For every pair of codewords, $\fc_{i_1}$ and $\fc_{i_2}$, such that $i_1,i_2 \in [\![M]\!]$ and $i_1\neq i_2$, there exist $t \in [\![n]\!]$, such that,
\begin{align}
    \label{Ineq.Conv_Distance}
    \left| c_{i_1,t} - c_{i_2,t} \right| > \epsilon_n' \;,
\end{align}
where
\begin{align}
    \label{Eq.Conv_Epsilon_n}
    \epsilon_n' = \frac{P_{\,\text{max}}}{n^{1+b}} \;,
\end{align}
with $b > 0$ being an arbitrarily small constant.
\end{lemma}
In the following, we provide the proof of Lemma~\ref{Lem.Converse}. The method of proof is by contradiction, namely, we assume that the condition given in \eqref{Ineq.Conv_Distance} is violated and then we show that this leads to a contradiction, namely, sum of the type I and type II error probabilities converge to one, i.e., $\lim_{n\to\infty} \left[ P_{e,1}(i_1) + P_{e,2}(i_2,i_1) \right] = 1$.
\begin{proof}
Fix $e_1$ and $e_2$. Let $\mu,\theta,\eta,\zeta$ be arbitrarily small positive. Assume to the contrary that there exist two messages $i_1$ and $i_2$, where $i_1\neq i_2$, such that,
\begin{align}
    \label{Ineq.Conv_Distance_Neg}
    \left| c_{i_1,t} - c_{i_2,t} \right| \leq \epsilon_n' \;,
\end{align}
which implies
\begin{align}
    \label{Ineq.Conv_Distance_Neg-1}
    c_{i_1,t} - c_{i_2,t} & \geq - \epsilon_n' \;,
    \\
    \label{Ineq.Conv_Distance_Neg-2}
    c_{i_2,t} - c_{i_1,t} & \geq - \epsilon_n' \;,
    \\
    \label{Ineq.Conv_Distance_Neg-3}
    c_{i_1,t} - c_{i_2,t} & \leq \epsilon_n' \;,
    \\
    \label{Ineq.Conv_Distance_Neg-4}
    c_{i_2,t} - c_{i_1,t} & \leq \epsilon_n' \;.
\end{align}
Observe that
\begin{align}
    \label{Eq.Error_Sum_1}
    P_{e,1}(i_1) + P_{e,2}(i_2,i_1)
    & = \left[ 1 - \sum_{\fy \in \mathbbmss{T}_{i_1}} W^n \big( \fy \, \big| \, \fc_{i_1} \big) \right] + \sum_{\fy \in \mathbbmss{T}_{i_1}} W^n \big( \fy \, \big| \, \fc_{i_2} \big)
    \,.\,
\end{align}
\begin{align}
    \label{Ineq.Cond_Channel_Diff}
    & W^n \big( \fy \, \big| \, \fc_{i_1} \big) - W^n \big( \fy \, \big| \, \fc_{i_2} \big)
    \nonumber\\
    & = W^n \big( \fy \, \big| \, \fc_{i_1} \big) \left[ 1 - \frac{W^n \big( \fy \, \big| \, \fc_{i_2} \big)}{W^n \big( \fy \, \big| \, \fc_{i_1} \big)}\right]
    \nonumber\\&
    = W^n \big( \fy \, \big| \, \fc_{i_1} \big) \left[ 1 - \frac{\prod_{t=1}^n \binom{T_{\rm s}c_{i_2,t}}{y_t} p^{y_t} (1-p)^{T_{\rm s}c_{i_2,t}-y_t}}{\prod_{t=1}^n \binom{T_{\rm s}c_{i_1,t}}{y_t} p^{y_t} (1-p)^{T_{\rm s}c_{i_1,t}-y_t}}\right]
    \nonumber\\&
    = W^n \big( \fy \, \big| \, \fc_{i_1} \big) \left[ 1 - \prod_{t=1}^n \frac{\binom{T_{\rm s}c_{i_2,t}}{y_t}}{\binom{T_{\rm s}c_{i_2,t}}{y_t}} \cdot (1-p)^{T_{\rm s}\big(c_{i_2,t}-c_{i_1,t}\big)} \right]
    \nonumber\\&
    = W^n \big( \fy \, \big| \, \fc_{i_1} \big) \left[ 1 - \prod_{t=1}^n \frac{T_{\rm s}c_{i_2,t} !}{T_{\rm s}c_{i_1,t} !} \cdot \frac{(T_{\rm s}c_{i_1,t} - y_t)!}{(T_{\rm s}c_{i_2,t} - y_t)!}
    \cdot (1-p)^{T_{\rm s}\big(c_{i_2,t}-c_{i_1,t}\big)} \right] \,.
\end{align}
In order to bound \eqref{Ineq.Cond_Channel_Diff}, we exploit the following useful double-inequality regarding ratio of two Gamma functions \cite[Eq.~4.15]{Qi09}. For $0 < a < b$, we have
\begin{align}
    \min\left\{ a,\frac{a+b-1}{2} \right\} \leq \left( \frac{\Gamma(a)}{\Gamma(b)} \right)^{\frac{1}{a-b}} \leq \max\left\{ a,\frac{a+b-1}{2} \right\} \,.
\end{align}
To analyze more accurately, we divide into three cases.
\begin{itemize}
    \setlength\itemsep{.5em}
    \item \textbf{Case 1:} Where $c_{i_1,t} < c_{i_2,t} \,, \forall t \in [\![n]\!]$
    \item \textbf{Case 2:} Where $c_{i_2,t} < c_{i_1,t} \,, \forall t \in [\![n]\!]$
    \item \textbf{Case 3:}
    $\begin{cases}
        c_{i_1,t} < c_{i_2,t} & \mbox{ for ${n_1}_{\geq 1}$ indices}\\
        c_{i_2,t} < c_{i_1,t} & \mbox{ for ${n_2}_{\geq 1}$ indices}\\
    \end{cases}$ \,, $n_1 + n_2 = n$
\end{itemize}
\subsection{Case 1}
Consider the case 1, i.e, where $c_{i_1,t} < c_{i_2,t} \,, \forall t \in [\![n]\!]$. Then, we set $a = T_{\rm s}c_{i_1,t} + 1$ and $b = T_{\rm s}c_{i_2,t} + 1$. Now, condition $0 < a < b$ is met and we obtain
\begin{align}
    \frac{T_{\rm s}c_{i_2,t} !}{T_{\rm s}c_{i_1,t} !} & = \frac{\Gamma(T_{\rm s}c_{i_2,t} + 1)}{\Gamma(T_{\rm s}c_{i_1,t} + 1)}
    \nonumber\\&
    = \frac{\Gamma(b)}{\Gamma(a)}
    \nonumber\\&
    \geq \left( \frac{1}{\max\left\{ a,\frac{a+b-1}{2} \right\}} \right)^{a-b}
    \nonumber\\&
    = \left( \frac{1}{\max\left\{ T_{\rm s}c_{i_1,t} - y_t + 1 ,  \frac{T_{\rm s} (c_{i_1,t} + c_{i_2,t}) - 2y_t + 1}{2} \right\}} \right)^{T_{\rm s} (c_{i_1,t} - c_{i_2,t})}
    \nonumber\\&
    \geq \left( \frac{1}{\max\left\{ AT_s + 1 , AT_s + \frac{1}{2} \right\}} \right)^{T_{\rm s} (c_{i_1,t} - c_{i_2,t})}
    \nonumber\\&
    \geq \left( \frac{1}{AT_s + 1} \right)^{T_{\rm s} (c_{i_1,t} - c_{i_2,t})} \,.
\end{align}
and
\begin{align}
    \frac{T_{\rm s}c_{i_2,t} !}{T_{\rm s}c_{i_1,t} !} & = \frac{\Gamma(T_{\rm s}c_{i_2,t} + 1)}{\Gamma(T_{\rm s}c_{i_1,t} + 1)}
    \nonumber\\&
    = \frac{\Gamma(b)}{\Gamma(a)}
    \nonumber\\&
    \leq \left( \frac{1}{\min\left\{ a,\frac{a+b-1}{2} \right\}} \right)^{a-b}
    \nonumber\\&
    = \left( \frac{1}{\min\left\{ T_{\rm s}c_{i_1,t} - y_t + 1 ,  \frac{T_{\rm s} (c_{i_1,t} + c_{i_2,t}) - 2y_t + 1}{2} \right\}} \right)^{T_{\rm s} (c_{i_1,t} - c_{i_2,t})}
    \nonumber\\&
    \leq \left( \frac{1}{\frac{1}{2}} \right)^{T_{\rm s} (c_{i_1,t} - c_{i_2,t})}
    \nonumber\\&
    \leq 2^{T_{\rm s} (c_{i_1,t} - c_{i_2,t})}
    \,.
\end{align}
Second, we set $a = T_{\rm s}c_{i_1,t} - y_t + 1$ and $b = T_{\rm s}c_{i_2,t} - y_t + 1$. Now, again condition $0 < a < b$ is met and we obtain
\begin{align}
    \frac{(T_{\rm s}c_{i_1,t} - y_t)!}{(T_{\rm s}c_{i_2,t} - y_t)!} & = \frac{\Gamma( T_{\rm s}c_{i_1,t} - y_t + 1)}{\Gamma( T_{\rm s}c_{i_2,t} - y_t + 1)}
    \nonumber\\&
    = \frac{\Gamma(a)}{\Gamma(b)}
    \nonumber\\&
    \geq \left( \min\left\{ a,\frac{a+b-1}{2} \right\} \right)^{a-b}
    \nonumber\\&
    = \left( \min\left\{ T_{\rm s}c_{i_2,t} - y_t + 1 ,  \frac{T_{\rm s} (c_{i_1,t} + c_{i_2,t}) - 2y_t + 1}{2} \right\} \right)^{T_{\rm s} (c_{i_1,t} - c_{i_2,t})}
    \nonumber\\&
    \geq \left( \frac{1}{2} \right)^{T_{\rm s} (c_{i_1,t} - c_{i_2,t})} \,.
\end{align}
and
\begin{align}
    \frac{(T_{\rm s}c_{i_1,t} - y_t)!}{(T_{\rm s}c_{i_2,t} - y_t)!} & = \frac{\Gamma( T_{\rm s}c_{i_1,t} - y_t + 1)}{\Gamma( T_{\rm s}c_{i_2,t} - y_t + 1)}
    \nonumber\\&
    = \frac{\Gamma(a)}{\Gamma(b)}
    \nonumber\\&
    \leq \left( \max\left\{ a,\frac{a+b-1}{2} \right\} \right)^{a-b}
    \nonumber\\&
    = \left( \max\left\{ T_{\rm s}c_{i_1,t} - y_t + 1 , \frac{T_{\rm s} (c_{i_1,t} + c_{i_2,t}) - 2y_t + 1}{2} \right\} \right)^{T_{\rm s} (c_{i_1,t} - c_{i_2,t})}
    \nonumber\\&
    = \left( \max\left\{ T_{\rm s}c_{i_1,t} + 1 , \frac{T_{\rm s} (c_{i_1,t} + c_{i_2,t}) + 1}{2} \right\} \right)^{T_{\rm s} (c_{i_1,t} - c_{i_2,t})}
    \nonumber\\&
    = \left( \max\left\{ AT_{\rm s} + 1 , \frac{2AT_{\rm s} + 1}{2} \right\} \right)^{T_{\rm s} (c_{i_1,t} - c_{i_2,t})}
    \nonumber\\&
    \leq \left( AT_{\rm s} + 1 \right)^{T_{\rm s} (c_{i_1,t} - c_{i_2,t})} \,.
\end{align}
Now, observe that
\begin{align}
    & \prod_{t=1}^n \frac{T_{\rm s}c_{i_2,t} !}{T_{\rm s}c_{i_1,t} !} \cdot \frac{(T_{\rm s}c_{i_1,t} - y_t)!}{(T_{\rm s}c_{i_2,t} - y_t)!} \cdot (1-p)^{T_{\rm s}\big(c_{i_2,t}-c_{i_1,t}\big)}
    \nonumber\\&
    \geq \prod_{t=1}^n \left( \frac{1}{AT_s + 1} \right)^{T_{\rm s} (c_{i_1,t} - c_{i_2,t})} \cdot \left( \frac{1}{2} \right)^{T_{\rm s} (c_{i_1,t} - c_{i_2,t})} \cdot (1-p)^{T_{\rm s}\big(c_{i_2,t}-c_{i_1,t}\big)} 
    \nonumber\\&
    = \left( \frac{1}{2(AT_s + 1)} \right)^{\sum_{t=1}^n T_{\rm s} (c_{i_1,t} - c_{i_2,t})} \cdot \left( \frac{1}{1-p} \right)^{\sum_{t=1}^n T_{\rm s} (c_{i_1,t} - c_{i_2,t})}
    \nonumber\\&
    = \left( 1 - \frac{2AT_s + 1}{2(AT_s + 1)} \right)^{\sum_{t=1}^n T_{\rm s} (c_{i_1,t} - c_{i_2,t})} \cdot \left( \frac{1}{1-p} \right)^{\sum_{t=1}^n T_{\rm s} (c_{i_1,t} - c_{i_2,t})}
    \nonumber\\&
    \stackrel{(a)}{\geq} \left( 1 - \frac{2AT_s + 1}{2(AT_s + 1)} \right)^{-\sum_{t=1}^n T_s \epsilon_n'} \cdot \left( \frac{1}{1-p} \right)^{-\sum_{t=1}^n T_s \epsilon_n'}
    \nonumber\\&
    = \left( \frac{1}{\left( 1 - \frac{2AT_s + 1}{2(AT_s + 1)}\right)^{\sum_{t=1}^n T_s \epsilon_n'}} \right) \cdot \left( \left( 1-p \right)^{\sum_{t=1}^n T_s \epsilon_n'} \right)
    \nonumber\\&
    \stackrel{(b)}{\geq} \left( \frac{1}{\left( 1 - \frac{2AT_s + 1}{2(AT_s + 1)} \cdot \sum_{t=1}^n T_s \epsilon_n' \right)} \right) \cdot \left( \left( 1 - p \cdot \sum_{t=1}^n T_s \epsilon_n' \right) \right)
    \nonumber\\&
    \stackrel{(c)}{\geq} \left( \frac{1}{\left( 1 - \frac{2AT_s + 1}{2(AT_s + 1)} \cdot \sum_{t=1}^n T_s \epsilon_n' \right)} \right) \cdot \left( 1 - p T_{\rm s} n\epsilon_n' \right)
    \nonumber\\&
    \stackrel{(d)}{\geq} \left( \frac{1}{\left( 1 - \frac{2AT_s + 1}{2(AT_s + 1)} \cdot \sum_{t=1}^n T_s \epsilon_n' \right)} \right) \cdot \left( 1 - \frac{p T_{\rm s} P_{\,\text{max}}}{n^b} \right)
    \nonumber\\&
    \stackrel{(e)}{\geq} \frac{1}{1 - 0} \cdot (1 - \kappa)
    \nonumber\\&
    \geq
    1 - \kappa \,.
\end{align}
where $(a)$ holds by \eqref{Ineq.Conv_Distance_Neg-1}, $(b)$ follows from the Bernoulli inequality, i.e., $(1 - x)^r \leq 1 + rx$ for $x \geq -1$ and $0 \leq \forall r_{\in \mathbb{R}} \leq 1$, $(c)$ holds since $\sum_{t=1}^n T_s \epsilon_n' = T_s n \epsilon_n'$, $(d)$ holds by \eqref{Eq.Conv_Epsilon_n}, and $(e)$ follows from $\sum_{t=1}^n T_s \epsilon_n' \geq 0$.

On the other hand, to provide an upper bound, we have
\begin{align}
    & \prod_{t=1}^n \frac{T_{\rm s}c_{i_2,t} !}{T_{\rm s}c_{i_1,t} !} \cdot \frac{(T_{\rm s}c_{i_1,t} - y_t)!}{(T_{\rm s}c_{i_2,t} - y_t)!} \cdot (1-p)^{T_{\rm s}\big(c_{i_2,t}-c_{i_1,t}\big)}
    \nonumber\\&
    \stackrel{(a)}{\leq} \prod_{t=1}^n 2^{T_{\rm s} (c_{i_1,t} - c_{i_2,t})} \cdot \left( AT_{\rm s} + 1 \right)^{T_{\rm s} (c_{i_1,t} - c_{i_2,t})}
    \nonumber\\&
    = 2^{\sum_{t=1}^n T_{\rm s} (c_{i_1,t} - c_{i_2,t})} \cdot \left( AT_{\rm s} + 1 \right)^{\sum_{t=1}^n T_{\rm s} (c_{i_1,t} - c_{i_2,t})}
    \nonumber\\&
    \stackrel{(b)}{\leq} 2^{T_{\rm s} n\epsilon_n'} \cdot \left( AT_{\rm s} + 1 \right)^{T_{\rm s} n\epsilon_n'}
    \nonumber\\&
    = \left( 2(AT_{\rm s} + 1) \right)^{T_{\rm s} n\epsilon_n'}
    \nonumber\\&
    \leq \left( 1 + 2AT_{\rm s} \right)^{T_{\rm s}P_{\,\text{max}} / n^b}
    \nonumber\\&
    \stackrel{(c)}{\leq}  1 + 2AT_{\rm s} \cdot T_{\rm s}P_{\,\text{max}} / n^b
    \nonumber\\&
    \leq 1 + \kappa \,.
\end{align}
where $(a)$ holds since $1 - p \leq 1$, $(b)$ follows from \eqref{Ineq.Conv_Distance_Neg-4}, and $(c)$ holds by the Bernoulli inequality, i.e., $(1 - x)^r \leq 1 + rx$ for $x \geq -1$ and $0 \leq \forall r_{\in \mathbb{R}} \leq 1$.

\subsection{Case 2}
Consider the case 2, i.e, where $c_{i_2,t} < c_{i_1,t} \,, \forall t \in [\![n]\!]$. Then, we set $a = T_{\rm s}c_{i_2,t} + 1$ and $b = T_{\rm s}c_{i_1,t} + 1$. Now, condition $0 < a < b$ is met and we obtain
\begin{align}
    \frac{T_{\rm s}c_{i_2,t} !}{T_{\rm s}c_{i_1,t} !} & = \frac{\Gamma(T_{\rm s}c_{i_2,t} + 1)}{\Gamma(T_{\rm s}c_{i_1,t} + 1)}
    \nonumber\\&
    = \frac{\Gamma(a)}{\Gamma(b)}
    \nonumber\\&
    \geq \left( \min\left\{ a,\frac{a+b-1}{2} \right\} \right)^{a-b}
    \nonumber\\&
    = \left( \min\left\{ T_{\rm s}c_{i_2,t} + 1 , \frac{T_{\rm s} (c_{i_2,t} + c_{i_1,t}) + 1}{2} \right\} \right)^{T_{\rm s} (c_{i_2,t} - c_{i_1,t})}
    \nonumber\\&
    \geq \left( \frac{1}{2} \right)^{T_{\rm s} (c_{i_2,t} - c_{i_1,t})} \,.
\end{align}
and
\begin{align}
    \frac{T_{\rm s}c_{i_2,t} !}{T_{\rm s}c_{i_1,t} !} & = \frac{\Gamma(T_{\rm s}c_{i_2,t} + 1)}{\Gamma(T_{\rm s}c_{i_1,t} + 1)}
    \nonumber\\&
    = \frac{\Gamma(a)}{\Gamma(b)}
    \nonumber\\&
    \leq \left( \max\left\{ a,\frac{a+b-1}{2} \right\} \right)^{a-b}
    \nonumber\\&
    = \left( \max\left\{ T_{\rm s}c_{i_2,t} + 1 , \frac{T_{\rm s} (c_{i_2,t} + c_{i_1,t}) - 2y_t + 1}{2} \right\} \right)^{T_{\rm s} (c_{i_2,t} - c_{i_1,t})}
    \nonumber\\&
    = \left( \max\left\{ T_{\rm s}c_{i_2,t} + 1 , \frac{T_{\rm s} (c_{i_2,t} + c_{i_1,t}) + 1}{2} \right\} \right)^{T_{\rm s} (c_{i_2,t} - c_{i_1,t})}
    \nonumber\\&
    \leq \left( \max\left\{ AT_{\rm s} + 1 , \frac{2AT_{\rm s} + 1}{2} \right\} \right)^{T_{\rm s} (c_{i_2,t} - c_{i_1,t})}
    \nonumber\\&
    \leq \left( AT_{\rm s} + 1 \right)^{T_{\rm s} (c_{i_2,t} - c_{i_1,t})} \,.
\end{align}
Second, we set $a = T_{\rm s}c_{i_2,t} - y_t + 1$ and $b = T_{\rm s}c_{i_1,t} - y_t + 1$. Now, again condition $0 < a < b$ is met and we obtain
\begin{align}
    \frac{(T_{\rm s}c_{i_1,t} - y_t)!}{(T_{\rm s}c_{i_2,t} - y_t)!} & = \frac{\Gamma( T_{\rm s}c_{i_1,t} - y_t + 1)}{\Gamma( T_{\rm s}c_{i_2,t} - y_t + 1)}
    \nonumber\\&
    = \frac{\Gamma(b)}{\Gamma(a)}
    \nonumber\\&
    \geq \left( \frac{1}{\max\left\{ a,\frac{a+b-1}{2} \right\}} \right)^{a-b}
    \nonumber\\&
    = \left( \frac{1}{\max\left\{ T_{\rm s}c_{i_2,t} - y_t + 1 ,  \frac{T_{\rm s} (c_{i_1,t} + c_{i_2,t}) - 2y_t + 1}{2} \right\}} \right)^{T_{\rm s} (c_{i_2,t} - c_{i_1,t})}
    \nonumber\\&
    \geq \left( \frac{1}{\max\left\{ AT_s + 1 , AT_s + \frac{1}{2} \right\}} \right)^{T_{\rm s} (c_{i_2,t} - c_{i_1,t})}
    \nonumber\\&
    \geq \left( \frac{1}{AT_s + 1} \right)^{T_{\rm s} (c_{i_2,t} - c_{i_1,t})} \,.
\end{align}
and
\begin{align}
    \frac{(T_{\rm s}c_{i_1,t} - y_t)!}{(T_{\rm s}c_{i_2,t} - y_t)!} & = \frac{\Gamma( T_{\rm s}c_{i_1,t} - y_t + 1)}{\Gamma( T_{\rm s}c_{i_2,t} - y_t + 1)}
    \nonumber\\&
    = \frac{\Gamma(b)}{\Gamma(a)}
    \nonumber\\&
    \leq \left( \frac{1}{\min\left\{ a,\frac{a+b-1}{2} \right\}} \right)^{a-b}
    \nonumber\\&
    = \left( \frac{1}{\min\left\{ T_{\rm s}c_{i_2,t} - y_t + 1 ,  \frac{T_{\rm s} (c_{i_1,t} + c_{i_2,t}) - 2y_t + 1}{2} \right\}} \right)^{T_{\rm s} (c_{i_2,t} - c_{i_1,t})}
    \nonumber\\&
    \leq \left( \frac{1}{\frac{1}{2}} \right)^{T_{\rm s} (c_{i_2,t} - c_{i_1,t})}
    \nonumber\\&
    \leq 2^{T_{\rm s} (c_{i_2,t} - c_{i_1,t})} \,.
\end{align}
Now, observe that
\begin{align}
    & \prod_{t=1}^n \frac{T_{\rm s}c_{i_2,t} !}{T_{\rm s}c_{i_1,t} !} \cdot \frac{(T_{\rm s}c_{i_1,t} - y_t)!}{(T_{\rm s}c_{i_2,t} - y_t)!} \cdot (1-p)^{T_{\rm s}\big(c_{i_2,t}-c_{i_1,t}\big)}
    \nonumber\\&
    \geq \prod_{t=1}^n \left( \frac{1}{2} \right)^{T_{\rm s} (c_{i_2,t} - c_{i_1,t})} \cdot \left( \frac{1}{AT_s + 1} \right)^{T_{\rm s} (c_{i_2,t} - c_{i_1,t})} \cdot (1-p)^{T_{\rm s}\big(c_{i_2,t}-c_{i_1,t}\big)}
    \nonumber\\&
    = \left( \frac{1}{2(AT_s + 1)} \right)^{\sum_{t=1}^n T_{\rm s} (c_{i_2,t} - c_{i_1,t})} \cdot \left( 1-p \right)^{\sum_{t=1}^n T_{\rm s} (c_{i_2,t} - c_{i_1,t})}
    \nonumber\\&
    = \left( 1 - \frac{2AT_s + 1}{2(AT_s + 1)} \right)^{\sum_{t=1}^n T_{\rm s} (c_{i_2,t} - c_{i_1,t})} \cdot \left( 1-p \right)^{\sum_{t=1}^n T_{\rm s} (c_{i_2,t} - c_{i_1,t})}
    \nonumber\\&
    \stackrel{(a)}{\geq} \left( 1 - \frac{2AT_s + 1}{2(AT_s + 1)} \right)^{-\sum_{t=1}^n T_s\epsilon_n'} \cdot \left( 1-p \right)^{-\sum_{t=1}^n T_s\epsilon_n'}
    \nonumber\\&
    = \left( \frac{1}{\left( 1 - \frac{2AT_s + 1}{2(AT_s + 1)}\right)^{\sum_{t=1}^n T_s\epsilon_n'}} \right) \cdot \left( \frac{1}{\left( 1-p \right)^{\sum_{t=1}^n T_s\epsilon_n'}} \right)
    \nonumber\\&
    \stackrel{(b)}{\geq} \left( \frac{1}{\left( 1 - \frac{2AT_s + 1}{2(AT_s + 1)} \cdot \sum_{t=1}^n T_s\epsilon_n' \right)} \right) \cdot \left( \frac{1}{\left( 1 - p \cdot \sum_{t=1}^n T_s\epsilon_n' \right)} \right)
    \nonumber\\&
    \stackrel{(c)}{\geq} \frac{1}{1 - 0} \cdot \frac{1}{1 - 0}
    \nonumber\\&
    = 1
    \nonumber\\&
    \geq
    1 - \kappa \,.
\end{align}
where $(a)$ holds by \eqref{Ineq.Conv_Distance_Neg-2}, $(b)$ follows from the Bernoulli inequality, i.e., $(1 - x)^r \leq 1 + rx$ for $x \geq -1$ and $0 \leq \forall r_{\in \mathbb{R}} \leq 1$, and $(c)$ holds since $\sum_{t=1}^n T_s \epsilon_n' \geq 0$.

On the other hand, to provide an upper bound, we have
\begin{align}
    & \prod_{t=1}^n \frac{T_{\rm s}c_{i_2,t} !}{T_{\rm s}c_{i_1,t} !} \cdot \frac{(T_{\rm s}c_{i_1,t} - y_t)!}{(T_{\rm s}c_{i_2,t} - y_t)!} \cdot (1-p)^{T_{\rm s}\big(c_{i_2,t}-c_{i_1,t}\big)}
    \nonumber\\&
    \stackrel{(a)}{\leq} \prod_{t=1}^n \left( AT_{\rm s} + 1 \right)^{T_{\rm s} (c_{i_2,t} - c_{i_1,t})} \cdot 2^{T_{\rm s} (c_{i_2,t} - c_{i_1,t})}
    \nonumber\\&
    = \left( AT_{\rm s} + 1 \right)^{\sum_{t=1}^n T_{\rm s} (c_{i_2,t} - c_{i_1,t})} \cdot 2^{\sum_{t=1}^n T_{\rm s} (c_{i_2,t} - c_{i_1,t})}
    \nonumber\\&
    \stackrel{(b)}{\leq} \left( AT_{\rm s} + 1 \right)^{T_{\rm s} n\epsilon_n'} \cdot 2^{T_{\rm s} n\epsilon_n'}
    \nonumber\\&
    = \left( 2(AT_{\rm s} + 1) \right)^{T_{\rm s} n\epsilon_n'}
    \nonumber\\&
    \leq \left( 1 + 2AT_{\rm s} \right)^{T_{\rm s}P_{\,\text{max}} / n^b}
    \nonumber\\&
    \stackrel{(c)}{\leq}  1 + 2AT_{\rm s} \cdot T_{\rm s}P_{\,\text{max}} / n^b
    \nonumber\\&
    \leq 1 + \kappa \,.
\end{align}
where $(a)$ holds since $1 - p \leq 1$, $(b)$ follows from \eqref{Ineq.Conv_Distance_Neg-4}, and $(c)$ holds by the Bernoulli inequality, i.e., $(1 - x)^r \leq 1 + rx$ for $x \geq -1$ and $0 \leq \forall r_{\in \mathbb{R}} \leq 1$.

\subsection{Case 3}
Consider the case 3, i.e., where for $n_1$ indices we have $c_{i_1,t} < c_{i_2,t}$ and for $n_2$ indices we have $c_{i_2,t} < c_{i_1,t}$ such that the total indices sum up to $n$, i.e., $n_1 + n_2 = n$. Therefore, for $n_1$ indices the rule (corresponding inequalities) of case 1 applies and for $n_2$ indices the rule (corresponding inequalities) of case 2 applies. That is, we have (follow the continuation in the landscape page in below)
\afterpage{
\begin{landscape}
\begin{align}
    & \prod_{t=1}^n \frac{T_{\rm s}c_{i_2,t} !}{T_{\rm s}c_{i_1,t} !} \cdot \frac{(T_{\rm s}c_{i_1,t} - y_t)!}{(T_{\rm s}c_{i_2,t} - y_t)!} \cdot (1-p)^{T_{\rm s}\big(c_{i_2,t}-c_{i_1,t}\big)} \leq \left[ \prod_{t=1}^{n_1} 2^{T_{\rm s} (c_{i_1,t} - c_{i_2,t})} \cdot \prod_{t=1}^{n_2} \left( AT_{\rm s} + 1 \right)^{T_{\rm s} (c_{i_2,t} - c_{i_1,t})} \right]
    \nonumber\\&
    \cdot \left[ \prod_{t=1}^{n_1} \left( AT_{\rm s} + 1 \right)^{T_{\rm s} (c_{i_1,t} - c_{i_2,t})} \cdot \prod_{t=1}^{n_2} 2^{T_{\rm s} (c_{i_2,t} - c_{i_1,t})} \right] \cdot \left[ \prod_{t=1}^{n_1} (1-p)^{T_{\rm s}\big(c_{i_2,t}-c_{i_1,t}\big)} \cdot \prod_{t=1}^{n_2} (1-p)^{T_{\rm s}\big(c_{i_2,t}-c_{i_1,t}\big)} \right]
    \nonumber\\&
    \leq \left[ \prod_{t=1}^{n_1} 2^{T_{\rm s} (c_{i_1,t} - c_{i_2,t})} \cdot \left( AT_{\rm s} + 1 \right)^{T_{\rm s} (c_{i_1,t} - c_{i_2,t})} \cdot (1-p)^{T_{\rm s}\big(c_{i_2,t}-c_{i_1,t}\big)} \right] \cdot
    \left[ \prod_{t=1}^{n_2} \left( AT_{\rm s} + 1 \right)^{T_{\rm s} (c_{i_2,t} - c_{i_1,t})} \cdot 2^{T_{\rm s} (c_{i_2,t} - c_{i_1,t})} \cdot (1-p)^{T_{\rm s}\big(c_{i_2,t}-c_{i_1,t}\big)} \right]
    \nonumber\\&
    = \left[ \prod_{t=1}^{n_1} \left( 2(AT_{\rm s} + 1) \right)^{T_{\rm s} (c_{i_1,t} - c_{i_2,t})} \cdot (1-p)^{T_{\rm s}\big(c_{i_2,t}-c_{i_1,t}\big)} \right] \cdot
    \left[ \prod_{t=1}^{n_2} \left( 2(AT_{\rm s} + 1) \right)^{T_{\rm s} (c_{i_2,t} - c_{i_1,t})} \cdot (1-p)^{T_{\rm s}\big(c_{i_2,t}-c_{i_1,t}\big)} \right]
    \nonumber\\&
    = \left[ \left( 2(AT_{\rm s} + 1) \right)^{\sum_{t=1}^{n_1} T_{\rm s} (c_{i_1,t} - c_{i_2,t})} \cdot (1-p)^{\sum_{t=1}^{n_1} T_{\rm s}\big(c_{i_2,t}-c_{i_1,t}\big)} \right] \cdot
    \left[ \left( 2(AT_{\rm s} + 1) \right)^{\sum_{t=1}^{n_2} T_{\rm s} (c_{i_2,t} - c_{i_1,t})} \cdot (1-p)^{\sum_{t=1}^{n_2} T_{\rm s}\big(c_{i_2,t}-c_{i_1,t}\big)} \right]
    \nonumber\\&
    = \left[ \left( 2(AT_{\rm s} + 1) \right)^{T_{\rm s} n_1\epsilon_n'} \cdot (1-p)^{T_{\rm s} n_1\epsilon_n'} \right] \cdot
    \left[ \left( 2(AT_{\rm s} + 1) \right)^{T_{\rm s} n_2\epsilon_n'} \cdot (1-p)^{T_{\rm s} n_2\epsilon_n'} \right]
    \nonumber\\&
    = \left[ \left( 2(AT_{\rm s} + 1) \right)^{T_{\rm s} (n_1 + n_2) \epsilon_n'} \cdot (1-p)^{T_{\rm s} (n_1 + n_2) \epsilon_n'} \right]
    \nonumber\\&
    = \left[ \left( 1 + 2AT_{\rm s} \right)^{T_{\rm s}P_{\,\text{max}} / n^b} \cdot (1-p)^{T_{\rm s}P_{\,\text{max}} / n^b} \right]
    \nonumber\\&
    \stackrel{(b)}{\leq} \left( 1 + 2AT_{\rm s} \right)^{T_{\rm s}P_{\,\text{max}} / n^b}
    \stackrel{(c)}{\leq} 1 + 2AT_{\rm s}\cdot T_{\rm s}P_{\,\text{max}} / n^b
    \nonumber\\&
    \leq 1 + \kappa \,.
\end{align}
where $(b)$ follows since $1 - p \leq 1$, and $(c)$ follows from the Bernoulli inequality, i.e., $(1 - x)^r \leq 1 + rx$ for $x \geq -1$ and $0 \leq \forall r_{\in \mathbb{R}} \leq 1$,
\end{landscape}
}

Therefore,
\begin{align}
    e_1 + e_2 & \geq P_{e,1}(i_1) + P_{e,2}(i_2,i_1) 
    \nonumber\\&
    \stackrel{(a)}{=} \left[ 1 - \sum_{\fy \in \mathbbmss{T}_{i_1}} W^n \big( \fy \, \big| \, \fc_{i_1} \big) \right] + \sum_{\fy \in \mathbbmss{T}_{i_1}} W^n \big( \fy \, \big| \, \fc_{i_2} \big)
    \nonumber\\&
    \stackrel{(b)}{=} 1 - \sum_{\fy \in \mathbbmss{T}_{i_1}} \left[ W^n \big( \fy \, \big| \, \fc_{i_1} \big) - W^n \big( \fy \, \big| \, \fc_{i_2} \big) \right]
    \nonumber\\&
    \stackrel{(c)}{=} 1 - \kappa \sum_{\fy \in \mathbbmss{T}_{i_1}} W^n \big( \fy \, \big| \, \fc_{i_1} \big)
    \nonumber\\&
    = 1 - \kappa \,,
\end{align}
where $(a)$ follows from, $(b)$ holds by \eqref{Ineq.Cond_Channel_Diff}, and $(c)$ follows since
\begin{align}
    \sum_{\fy \in \mathbbmss{T}_{i_1}} W^n \big( \fy \, \big| \, \fc_{i_1} \big) & = \Pr\left( \fy \in \mathbbmss{T}_{i_1} \right)
    \nonumber\\&
    \leq \Pr\left( \fy \in \mathbb{N}_0^n \right)
    \nonumber\\&
    = 1 \,.
\end{align}
Therefore, $e_1 + e_2 \geq 1 - \kappa$ which leads to a contradiction since for sufficiently small $\mu$ and vanishing error probabilities we obtain $\kappa < 1 - e_1 - e_2$. This completes the proof of Lemma~\ref{Lem.Converse}. 
\end{proof}
Next, we employ Lemma~\ref{Lem.Converse} to determine the upper bound on scale of codebook size. Notice that Lemma~\ref{Lem.Converse} implies that the distance between every pair of codewords fulfill
\begin{align}
    \label{Ineq.Conv_Distance_2}
    \norm{\fc_{i_1} - \fc_{i_2}} & \geq \left|c_{i_1,t} - c_{i_1,t} \right|
    \nonumber\\&
    \geq \epsilon_n'
    \nonumber\\&
    = \frac{P_{\,\text{max}}}{n^{1+b}}
    \;,
\end{align}
Thus, we can define an arrangement of non-overlapping spheres $\S_{\fc_i}(n,\epsilon'_n)$, i.e., spheres of radius $\epsilon'_n$ that are centered at the codewords $\fc_i$. Since the codewords all belong to a hyper cube $\Q_{\f0}(n,P_{\,\text{max}})$ with edge length $P_{\,\text{max}}$, it follows that the number of packed small spheres, i.e., the number of codewords $M$, is bounded by
\begin{align}
    \label{Eq.M}
    M & = \frac{\text{Vol}\left(\bigcup_{i=1}^{L}\S_{\fu_i}(n,r_0)\right)}{\text{Vol}(\S_{\fc_1}(n,r_0))}
    \nonumber\\&
    = \frac{\Updelta_n(\mathscr{S}) \cdot \text{Vol}\left(\Q_{\f0}(n,P_{\,\text{max}})\right)}{\text{Vol}(\S_{\fc_1}(n,r_0))}
    \nonumber\\&
    \leq 2^{-0.599n} \cdot\frac{P_{\,\text{max}}^n}{\text{Vol}(\S_{\fc_1}(n,r_0))} \;,\,
\end{align}
where the last inequality follows from inequality (\ref{Ineq.Density}). Thereby,
\begin{align}
    \label{Eq.Converse_Log_M}
    \log M & \leq \log \left( \frac{P_{\,\text{max}}^n}{\text{Vol}\left(\S_{\fc_1}(n,r_0)\right)} \right) - 0.599n
    \nonumber\\
    & = n \log P_{\,\text{max}} - n \log r_0 - n \log \sqrt{\pi} + \frac{1}{2}n\log \frac{n}{2} - \frac{n}{2}\log e + o(n) - 0.599n \,,\;
\end{align}
where the dominant term is again of order $n \log n$. Hence, for obtaining a finite value for the upper bound of the rate, $R$, \eqref{Eq.Converse_Log_M} induces the scaling law of $M$ to be $2^{(n\log n)R}$. Hence, by setting $M(n,R) = 2^{(n\log n)R}$ and $r_0 = \epsilon'_n = P_{\,\text{max}} / n^{1+b}$, we obtain
\begin{align}
    R & \leq \frac{1}{n\log n} \left[ n \log P_{\,\text{max}} - n \log r_0 - n \log \sqrt{\pi} + \frac{1}{2}n\log \frac{n}{2} - \frac{n}{2}\log e + o(n) - 0.599n \right]
    \nonumber\\
    & = \frac{1}{n\log n} \left[ \left( \frac{1}{2} + \left( 1 + b \right) \right) \, n \log n - n \left( \log P_{\,\text{max}} \sqrt{\pi e} + 1.099 \right)+ o(n) \right]
    \;,\,
\end{align}
which tends to $\frac{3}{2}$ as $n \to \infty$ and $b \to 0$. This completes the proof of Theorem~\ref{Th.DI-Capacity}.

\section{Summary and Future Directions}
\label{Sec.Summary}
In this work, we studied the DI problem over the Binomial channel. We assume that the transmitter is subject to both the average and peak molecule release rate constraints. Our results in this work may serve as a model for event-triggered based tasks in the context of future XG applications. In particular, we obtained lower and upper bounds on the DI capacity of the Binomial channel with the codebook size of $M(n,R)=2^{(n\log n)R}=n^{nR}$. Our results for the DI capacity of the Binomial channel revealed that the super-exponential scale of $n^{nR}=2^{(n\log n)R}$ is again the appropriate scale for codebook size. This scale coincides as of the codebook for the memoryless Gaussian channels \cite{Salariseddigh_IT,Salariseddigh_ITW} and Poisson channels \cite{Salariseddigh_GC_IEEE,Salariseddigh22} and stands considerably different from the traditional scales in transmission and RI setups where corresponding codebooks size grows exponentially and double exponentially, respectively.

We show the achievability proof using a sphere packing arrangement of hyper spheres and a distance decoder. In particular, we pack hyper spheres with radius $\sqrt{n\epsilon_n} \sim n^{\frac{1}{4}}$, inside a larger hyper sphere, which results in $\sim 2^{\frac{1}{4} n\log n}$ codewords. For the converse proof, we follow a similar approach as in our previous work for the DI over the Gaussian channels \cite{Salariseddigh_IT,Salariseddigh_arXiv_ITW}. That is, given a certain condition imposed on the codebook, using the continuity of the channel law, we reach to a contradiction on sum of the error probabilities. In general, the derivation here is more involved than the derivation in the DI case \cite{Salariseddigh_ITW} and entails employing of new analysis and inequalities. Here, proving the continuity of Binomial law requires dealing with binomial coefficients and factorial terms. We used inequalities on the ratio of two Gamma function depending on the relation of two codeword's symbols with each other in all possible cases. In our previous work on Gaussian channels with fading \cite{Salariseddigh_ITW}, the converse proof was based on establishing a minimum distance between Euclidean norm of each pair of codewords. Here, we consider a distance (absolute value norm) between symbols of two different codeword in the relevant Lemma; cf. \eqref{Ineq.Conv_Distance}.

\section*{Acknowledgements}
Salariseddigh was supported by the 6G-Life project under grant 16KISK002. Jamali’s work is supported in part by the LOEWE initiative (Hesse, Germany) within the emergenCITY center. Boche was supported in part by the German
Federal Ministry of Education and Research (BMBF) within the national initiative for Post Shannon Communication (NEWCOM) under Grant 16KIS1003K. Deppe was supported by the 6G-Life project under grant 16KISK002. Schober was supported by MAMOKO under grant 16KIS0913.

\section*{}
\bibliography{Lit}

\begin{thebibliography}{10}
\providecommand{\url}[1]{#1}
\csname url@samestyle\endcsname
\providecommand{\newblock}{\relax}
\providecommand{\bibinfo}[2]{#2}
\providecommand{\BIBentrySTDinterwordspacing}{\spaceskip=0pt\relax}
\providecommand{\BIBentryALTinterwordstretchfactor}{4}
\providecommand{\BIBentryALTinterwordspacing}{\spaceskip=\fontdimen2\font plus
\BIBentryALTinterwordstretchfactor\fontdimen3\font minus
  \fontdimen4\font\relax}
\providecommand{\BIBforeignlanguage}[2]{{%
\expandafter\ifx\csname l@#1\endcsname\relax
\typeout{** WARNING: IEEEtran.bst: No hyphenation pattern has been}%
\typeout{** loaded for the language `#1'. Using the pattern for}%
\typeout{** the default language instead.}%
\else
\language=\csname l@#1\endcsname
\fi
#2}}
\providecommand{\BIBdecl}{\relax}
\BIBdecl

\bibitem{FYECG16}
N.~Farsad, H.~B. Yilmaz, A.~Eckford, C.-B. Chae, and W.~Guo, ``A
  {C}omprehensive {S}urvey of {R}ecent {A}dvancements in {M}olecular
  {C}ommunication,'' \emph{IEEE Commun. Surveys Tuts.}, vol.~18, no.~3, pp.
  1887--1919, 2016.

\bibitem{Jamali19}
V.~Jamali, A.~Ahmadzadeh, W.~Wicke, A.~Noel, and R.~Schober, ``Channel
  {M}odeling {F}or {D}iffusive {M}olecular {C}ommunication - {A} {T}utorial
  {R}eview,'' \emph{Proc. IEEE}, vol. 107, no.~7, pp. 1256--1301, 2019.

\bibitem{Kuscu19}
M.~Kuscu, E.~Dinc, B.~A. Bilgin, H.~Ramezani, and O.~B. Akan, ``Transmitter and
  {R}eceiver {A}rchitectures {F}or {M}olecular {C}ommunications: {A} {S}urvey
  on {P}hysical {D}esign {W}ith {M}odulation, {C}oding, and {D}etection
  {T}echniques,'' \emph{Proc. IEEE}, vol. 107, no.~7, pp. 1302--1341, 2019.

\bibitem{Soldner20}
C.~A. S{\"o}ldner, E.~Socher, V.~Jamali, W.~Wicke, A.~Ahmadzadeh, H.-G.
  Breitinger, A.~Burkovski, K.~Castiglione, R.~Schober, and H.~Sticht, ``A
  {S}urvey of {B}iological {B}uilding {B}locks {F}or {S}ynthetic {M}olecular
  {C}ommunication {S}ystems,'' \emph{IEEE Commun. Surveys Tuts.}, vol.~22,
  no.~4, pp. 2765--2800, 2020.

\bibitem{Gohari16}
A.~Gohari, M.~Mirmohseni, and M.~Nasiri-Kenari, ``Information {T}heory of
  {M}olecular {C}ommunication:{D}irections and {C}hallenges,'' \emph{IEEE
  Trans. Mol. Biol. Multi-Scale Commun.}, vol.~2, no.~2, pp. 120--142, 2016.

\bibitem{Farsad17}
N.~Farsad, D.~Pan, and A.~Goldsmith, ``A {N}ovel {E}xperimental {P}latform
  {F}or in-vessel {M}ulti-{C}hemical {M}olecular {C}ommunications,'' in
  \emph{Proc. IEEE Global Commun. Conf.}, 2017, pp. 1--6.

\bibitem{6G_PST}
J.~A. Cabrera, H.~Boche, C.~Deppe, R.~F. Schaefer, C.~Scheunert, and F.~H.
  Fitzek, ``6{G} and {T}he {P}ost-{S}hannon {T}heory,'' in \emph{Shaping Future
  6G Networks: Needs, Impacts and Technologies}, N.~O. Frederiksen and
  H.~Gulliksen, Eds.\hskip 1em plus 0.5em minus 0.4em\relax Hoboken, NJ, United
  States: Wiley-Blackwell, 2021.

\bibitem{Schwenteck23}
P.~Schwenteck, G.~T. Nguyen, H.~Boche, W.~Kellerer, and F.~H.~P. Fitzek, ``6g
  {P}erspective of {M}obile {N}etwork {O}perators, {M}anufacturers, and
  {V}erticals,'' \emph{IEEE Networking Letters}, pp. 1--1, 2023.

\bibitem{S48}
C.~E. {Shannon}, ``A {M}athematical {T}heory of {C}ommunication,'' \emph{Bell
  Sys. Tech. J.}, vol.~27, no.~3, pp. 379--423, 1948.

\bibitem{AD89}
R.~{Ahlswede} and G.~{Dueck}, ``Identification {V}ia {C}hannels,'' \emph{IEEE
  Trans. Inf. Theory}, vol.~35, no.~1, pp. 15--29, 1989.

\bibitem{Salariseddigh22}
\BIBentryALTinterwordspacing
M.~J. Salariseddigh, U.~Pereg, H.~Boche, C.~Deppe, V.~Jamali, and R.~Schober,
  ``{D}eterministic {I}dentification {F}or {M}olecular {C}ommunications {O}ver
  {T}he {P}oisson {C}hannel,'' \emph{arXiv preprint arXiv:2203.02784}, 2022.
  [Online]. Available: \url{https://arxiv.org/abs/2203.02784}
\BIBentrySTDinterwordspacing

\bibitem{J85}
J.~J\'aJ\'a, ``Identification is {E}asier {T}han {D}ecoding,'' in \emph{Proc.
  Ann. Symp. Found. Comp. Scien.}, 1985, pp. 43--50.

\bibitem{AA20}
A.~Y. Anup~Rao, \emph{Communication Complexity: and Applications}.\hskip 1em
  plus 0.5em minus 0.4em\relax Cambridge University Press, 2020.

\bibitem{Salariseddigh_IT}
M.~J. Salariseddigh, U.~Pereg, H.~Boche, and C.~Deppe, ``Deterministic
  {I}dentification {O}ver {C}hannels {W}ith {P}ower {C}onstraints,'' \emph{IEEE
  Trans. Inf. Theory}, vol.~68, no.~1, pp. 1--24, 2022.

\bibitem{Salariseddigh_ICC}
------, ``Deterministic {I}dentification {O}ver {C}hannels {W}ith {P}ower
  {C}onstraints,'' in \emph{ICC 2021 - IEEE Intl. Conf. Commun.}, 2021, pp.
  1--6.

\bibitem{Salariseddigh_ITW}
------, ``Deterministic {I}dentification {O}ver {F}ading {C}hannels,'' in
  \emph{2020 IEEE Inf. Theory Workshop (ITW)}, 2021, pp. 1--5.

\bibitem{Salariseddigh_GC_IEEE}
M.~J. Salariseddigh, U.~Pereg, H.~Boche, C.~Deppe, and R.~Schober,
  ``Deterministic {I}dentification {O}ver {P}oisson {C}hannels,'' in \emph{2021
  IEEE Globecom Workshops (GC Wkshps)}, 2021, pp. 1--6.

\bibitem{Salariseddigh_GC_arXiv}
\BIBentryALTinterwordspacing
------, ``Deterministic {I}dentification {O}ver {P}oisson {C}hannels,''
  \emph{arXiv preprint arXiv:2107.06061}, 2021. [Online]. Available:
  \url{http://arxiv.org/abs/2107.06061.pdf}
\BIBentrySTDinterwordspacing

\bibitem{Salariseddigh22-3}
\BIBentryALTinterwordspacing
M.~J. Salariseddigh, M.~Spahovic, and C.~Deppe, ``Deterministic
  ${K}$-{I}dentification {F}or {S}low {F}ading {C}hannel,'' \emph{arXiv
  preprint arXiv:2212.02732, submitted to IEEE Inf. Theory Workshop 2023},
  2022. [Online]. Available: \url{https://arxiv.org/abs/2212.02732}
\BIBentrySTDinterwordspacing

\bibitem{Salariseddigh22_2}
\BIBentryALTinterwordspacing
M.~J. Salariseddigh, V.~Jamali, U.~Pereg, H.~Boche, C.~Deppe, and R.~Schober,
  ``Deterministic {I}dentification {F}or {MC} {ISI}-{P}oisson {C}hannel,''
  \emph{arXiv preprint arXiv:2211.11024, accepted for publication in IEEE Intl.
  Conf. Commun. 2023}, 2022. [Online]. Available:
  \url{http://arxiv.org/abs/2211.11024.pdf}
\BIBentrySTDinterwordspacing

\bibitem{Komninakis01}
C.~Komninakis, L.~Vandenberghe, and R.~D. Wesel, ``Capacity of {T}he {B}inomial
  {C}hannel, or {M}inimax {R}edundancy {F}or {M}emoryless {S}ources,'' in
  \emph{Intl. Symp. Inf. Theory}, 2001, pp. 127--127.

\bibitem{Wesel18}
R.~D. Wesel, E.~E. Wesel, L.~Vandenberghe, C.~Komninakis, and M.~Medard,
  ``Efficient {B}inomial {C}hannel {C}apacity {C}omputation {W}ith an
  {A}pplication to {M}olecular {C}ommunication,'' in \emph{2018 Inf. Theory
  Appl. Workshop (ITA)}, 2018, pp. 1--5.

\bibitem{Farsad17-2}
N.~Farsad, C.~Rose, M.~M{\'e}dard, and A.~Goldsmith, ``Capacity of {M}olecular
  {C}hannels {W}ith {I}mperfect {P}article-{I}ntensity {M}odulation and
  {D}etection,'' in \emph{2017 IEEE Intl. Symp. Inf. Theory (ISIT)}, 2017, pp.
  2468--2472.

\bibitem{Farsad20}
N.~Farsad, W.~Chuang, A.~Goldsmith, C.~Komninakis, M.~M{\'e}dard, C.~Rose,
  L.~Vandenberghe, E.~E. Wesel, and R.~D. Wesel, ``Capacities and {O}ptimal
  {I}nput {D}istributions {F}or {P}article-{I}ntensity {C}hannels,'' \emph{IEEE
  Trans. Mol. Biol. Multi-Scale Commun.}, vol.~6, no.~3, pp. 220--232, 2020.

\bibitem{Damrath20}
\BIBentryALTinterwordspacing
M.~Damrath, ``Channel {C}oding in {M}olecular {C}ommunication,'' Ph.D.
  dissertation, Technischen Fakultät der Christian-Albrechts-Universität zu
  Kiel, 2020. [Online]. Available:
  \url{https://macau.uni-kiel.de/servlets/MCRFileNodeServlet/macau_derivate_00001814/Martin_Damrath.pdf}
\BIBentrySTDinterwordspacing

\bibitem{Adam14}
A.~Noel, K.~C. Cheung, and R.~Schober, ``Improving {R}eceiver {P}erformance of
  {D}iffusive {M}olecular {C}ommunication {W}ith {E}nzymes,'' \emph{IEEE Trans.
  Nanobiosci.}, vol.~13, no.~1, pp. 31--43, 2014.

\bibitem{Farsad16-2}
N.~Farsad and A.~Goldsmith, ``A {M}olecular {C}ommunication {S}ystem {U}sing
  {A}cids, {B}ases and {H}ydrogen {I}ons,'' in \emph{2016 IEEE 17th Int.
  Workshop Signal Process. Adv. Wireless Commun. (SPAWC)}, 2016, pp. 1--6.

\bibitem{Jamali18}
V.~Jamali, N.~Farsad, R.~Schober, and A.~Goldsmith, ``Diffusive {M}olecular
  {C}ommunications {W}ith {R}eactive {M}olecules: {C}hannel {M}odeling and
  {S}ignal {D}esign,'' \emph{IEEE Trans. Mol. Biol. Multi-Scale Commun.},
  vol.~4, no.~3, pp. 171--188, 2018.

\bibitem{CHSN13}
J.~H. Conway and N.~J.~A. Sloane, \emph{Sphere {P}ackings, {L}attices and
  {G}roups}.\hskip 1em plus 0.5em minus 0.4em\relax Springer Science \&
  Business Media, 2013.

\bibitem{C10}
H.~Cohn, ``Order and {D}isorder in {E}nergy {M}inimization,'' in \emph{Proc.
  Int. Congr. Mathn.}\hskip 1em plus 0.5em minus 0.4em\relax World Scientific,
  2010, pp. 2416--2443.

\bibitem{F66}
W.~Feller, \emph{An Introduction to Probability Theory and its
  Applications}.\hskip 1em plus 0.5em minus 0.4em\relax John Wiley \& Sons,
  1966.

\bibitem{Qi09}
\BIBentryALTinterwordspacing
F.~Qi, ``Bounds {F}or {T}he {R}atio of {T}wo {G}amma {F}unctions--{F}rom
  {W}endel's and {R}elated {I}nequalities to {L}ogarithmically {C}ompletely
  {M}onotonic {F}unctions,'' \emph{arXiv preprint arXiv:0904.1048}, 2009.
  [Online]. Available: \url{https://arxiv.org/abs/0904.1048}
\BIBentrySTDinterwordspacing

\bibitem{Salariseddigh_arXiv_ITW}
\BIBentryALTinterwordspacing
M.~J. Salariseddigh \emph{et~al.}, ``Deterministic {I}dentification {O}ver
  {F}ading {C}hannels,'' \emph{arXiv preprint arXiv:2010.10010}, 2020.
  [Online]. Available: \url{https://arxiv.org/pdf/2010.10010.pdf}
\BIBentrySTDinterwordspacing

\end{thebibliography}
\end{document}